\documentclass[11pt]{article}

\usepackage{amsfonts,amsmath,amssymb,amsthm}
\usepackage{verbatim,float,url,enumerate}
\usepackage{caption}
\usepackage{graphicx,subcaption,epsfig,psfrag}
\usepackage{dsfont,bm,color,appendix}
\usepackage{centernot}
\usepackage{natbib}
\setcitestyle{numbers,square,comma,aysep={},yysep={;}}
\usepackage[final]{pdfpages}

\usepackage{chngcntr} %change Figure numbering in Appendix
\usepackage{fancyhdr}
\usepackage{tcolorbox}
\usepackage[margin=1in]{geometry}
\usepackage[lined, ruled, commentsnumbered]{algorithm2e}
\usepackage{setspace}
\usepackage{hyperref}

\newtheorem{theorem}{Theorem}

\newtheorem{proposition}{Proposition}
\newtheorem{definition}{Definition}

\newcommand{\citepeg}[1]{\citep[e.g.,][]{#1}}

\def\I{\mathbb{I}}

\def\wtilde{\widetilde}

\def\eqd{\stackrel{d}{=}}
\newcommand{\indep}{\ \rotatebox[origin=c]{90}{$\models$} \ }
\newcommand{\notindep}{\centernot{\indep}}

\def\maskf{\text{(mask,f)}}
\def\maskfi{\text{(mask,f)(i)}}
\def\maskm{\text{(mask,m)}}
\def\maskmi{\text{(mask,m)(i)}}
\def\i{\text{(i)}}
\def\fii{\text{f,(i)}}
\def\mii{\text{m,(i)}}
\def\B{\mathcal{B}}
\def\C{\mathcal{C}}
\def\Xk{\wtilde{X}}
\def\fixed{\mathcal{D}}

\usepackage{authblk} %% Better author affiliations

\title{Causal Inference in Genetic Trio Studies}
\author[1]{Stephen Bates}
\author[1]{Matteo Sesia}
\author[1,2]{Chiara Sabatti}
\author[1,3]{Emmanuel Cand\`es}
\affil[1]{Department of Statistics, Stanford University}
\affil[2]{Department of Biomedical Data Science, Stanford University}
\affil[3]{Department of Mathematics, Stanford University}
\date{\today}

\begin{document}

\maketitle

\begin{abstract}
We introduce a method to rigorously draw causal
inferences---inferences immune to all possible confounding---from
genetic data that include parents and offspring. Causal conclusions
are possible with these data because the natural randomness in meiosis
can be viewed as a high-dimensional randomized experiment. We make
this observation actionable by developing a novel conditional
independence test that identifies regions of the genome containing
distinct causal variants. The proposed {\em Digital Twin Test}
compares an observed offspring to carefully constructed synthetic
offspring from the same parents in order to determine statistical
significance, and it can leverage any black-box multivariate model and
additional non-trio genetic data in order to increase
power. Crucially, our inferences are based only on a well-established
mathematical description of the rearrangement of genetic material during meiosis
and make no assumptions about the relationship between the genotypes and phenotypes.
\end{abstract}


\section{Introduction}
Over the last 20 years, genome-wide association studies (GWAS) have established the association of thousands of genetic markers with phenotypes of interest, and these studies have the potential to shed light on essential scientific questions and to enable personalized medical care \citep{Visscher2017}. Analyzing GWAS data with increasingly large sample sizes poses new statistical challenges, however; in this new regime, any statistical association between a genetic variant and phenotype will now be detectable, including many irrelevant associations arising from non-genetic factors such as differing environmental conditions.
While there are existing methods to mitigate this problem \citep{Devlin1999, Price2006,kang2010}, such methods are not guaranteed to remove it entirely, so, with large sample sizes, they may leave many detectable associations that do not represent interesting biological activity.
As such, a fundamental challenge for GWAS in the next decade is to
move from detecting promising associations to rigorously establishing
causality. On that premise, the most trustworthy way to ascertain that
a statistical association is causal is to use a randomized experiment
\citep{rubin1974}, and parent-offspring duo or trio data record such an experiment in the sense that the locations of the recombination points during meiosis are randomized by nature. Building on this, we propose a method that analyzes the placement of such sites to provably report only biologically meaningful regions of the genome.

\subsection{Our contribution}
In this work, we first articulate how causal inference is possible in the trio design and then introduce the {\em Digital Twin Test}: a method for finding causal regions that is immune to confounding variables. Our contribution has four components:

\begin{enumerate}
\item {\bf Establishing causality in the trio design} We formalize the existing notion that family studies are immune to population structure, showing how to leverage the trio design to make causal inferences in a rigorous statistical sense. We formulate a broad class of tests for the null hypothesis of no causal effect, which includes existing methods such as the transmission disequilibrium test (TDT) \citep{Spielman1993, Spielman1996} and some of its variations for quantitative traits.

\item {\bf A method for discovering causal regions}
We introduce the Digital Twin Test to make causal inferences using trio data, improving on several aspects of the existing methods.

\begin{enumerate}[(a)]
\item {\bf Identifying distinct causal regions} 
Our method provably localizes causal variants within explicit windows along the genome, clearly showing the user when there are distinct causal effects. By contrast, although it is not widely known, the TDT is testing a less exact chromosome-wide null, so spurious findings arise from correlations among sites on the chromosome---see Section~\ref{sec:trio_ld} for an example.

\item {\bf Testing multiple hypotheses}
 Our method deals with multiple comparisons in a precise way, controlling either the family-wise error rate or the false discovery rate (FDR) without the need for a conservative Bonferroni correction. The heart of our solution is the creation of {\em independent} p-values for distinct regions of the genome, which can then be used together with more powerful multiple testing procedures.

\item {\bf Leveraging black-box models and subject matter knowledge} 
Our method increases power by incorporating any multivariate model and subject matter information. Critically, the error rate guarantees of the method do not rely whatsoever on the correctness of the prior information or of the phenotype model. 
\end{enumerate}
\end{enumerate}
While our inferences are based on trio data, our method can take advantage of additional case-control or population GWAS data to greatly increase power {\em while retaining the certified causal inferences.}  Since trio samples are harder to collect than case-control or population samples, the techniques we describe do not replace existing methods for the latter, but instead work together with them to rigorously establish that the detected associations are due to causal variants. Lastly, we highlight that our approach is flexible and naturally applies to binary, quantitative, or time-to-onset response variables.

\subsection{Trios, causality, and confounding in the literature}

We begin by outlining related methods from the genetics literature.
In inheritance, each parent transmits one complete set of DNA to the
progeny, copying large continuous segments from that parent's two
strands with occasional switches. Geneticists have long exploited the
randomness in this process to identify meaningful associations
\citep{hrr1981, Falk1987, Ott2011}. The launching point for this work
is the TDT \citep{Spielman1993, Spielman1996}, which checks if a given
allele is inherited more or less frequently than would be expected due
to chance among the observations with a disease. If the transmission
frequency deviates from the baseline frequency, the TDT reports an
association.  Beyond the original TDT, additional techniques for more
complex, partially observed pedigrees \citep{Thomson1995,
  Rabinowitz2000, Lange2002b} and utilizing multiple markers
\citep{lazzeroni1998conditional, Xu2007, Rakovski2008} have been developed; these are known as
{\em family-based association tests}. These methods are robust both to
modeling assumptions about the relationship between the trait and the
genotypes and to confounding due to {\em population structure},
namely, the presence of subpopulations with different allele
frequencies \citepeg{Laird2010}. Moreover, these techniques can be
extended to address quantitative traits \citep{Allison1997,
  Rabinowitz1997, Monks2000}, although,
unlike the TDT, some of these methods require parametric assumptions
about the relationship between the phenotypes and genotypes
\citep{Ewens2008}. Similarly, the method of pseudocontrols
\citep{Schaid1996, Cordell2002} computes the likelihood of an observed
transmission pattern of a small number of alleles based on a posited
model for the genotypes and phenotypes. To address the multiple
comparisons issue arising from looking at several variants at
  once, \citep{VanSteen2005} shows how to decouple the selection of
promising markers from the final construction of a p-value from
family-based association tests.

% % Linkage methods
% In parallel to the TDT and family-based association tests, {\em linkage} methods seek to identify causal regions by finding regions of inheritance in large pedigrees that are common to the members of a pedigree who have the disease \citep{Sobel1996, di2009conditional, CHEUNG2013504}. Such methods are similarly robust to population structure, but are primarily of use when studying {\em Mendelian diseases}, diseases that are caused by a single variant. In a different direction, recent work has used large-scale trio studies together with random effects models to determine the heritability of polygenic traits by isolating the direct genetic effects from environmental effects \citep{Kong2018, Young2018}. The idea driving these works is that some of the variation in a trait can be explained away by the parental genotypes, and only the remaining variability is a result of causal genetic effects.

Turning to the statistics literature, causal inference is concerned with understanding what would happen if one manipulated the system under study, such as if one changed the SNPs in a small region of the genome while leaving all other conditions the same. This was first explicitly formulated in the context of randomized experiments by Fisher and Neyman \citep{Fisher1925, splawa-neyman1990, fisher1935}, who
pointed out that randomization removes the effect of {\em
  confounders}, unmeasured variables that create dependence between
a covariate and the response even when directly changing the covariate would not affect the response. Therefore, any observed dependence in a randomized experiment must be due to a causal effect. 
This notion was later formalized for observational studies with the
potential outcomes framework \citep{rubin1974, rubin2005}  
%which posits that each observational unit has a fixed value for various possible treatment assignments, but the analyst only gets to see the one value corresponding to the treatment that was carried out. 
%where causal inference is formulated as a missing data problem in which the analyst seeks to understand the outcome under unseen configurations of the treatment assignment. 
and, alternatively, with structural equation models \citep{Bollen1989, Verma1990EquivalenceAS, pearl2009}. 
%At the same time, an alternative formalization was developed that uses structural equation models to describe the distribution of the a set of random variables in the presence of {\em interventions}, external manipulations that change only a subset of variables \citep{Bollen1989, Verma1990EquivalenceAS, pearl2009}. 
While there is a spirited debate in the research community about when each is more fruitful, in this work we will demonstrate that our causal inferences are valid simultaneously in both frameworks.

%conditional testing + causal discovery
Within causal inference, the problem of learning the structure of the true underlying model from data is known as {\em causal discovery} \citepeg{Spirtes2000,Kalisch2007,maathuis2009}. General methods for causal discovery exist, although they typically require a large number of conditional independence tests and the assumption that such tests can be carried out without any statistical error \citep{Chickering2002, He2008}. As a result, finite-sample results are rare. This work also uses conditional independence testing as the foundation for causal discovery but builds upon the conditional randomization test \citep{candes2016} to give finite-sample statistical guarantees. Our approach is also related to that of knockoffs \citep{barber2015, candes2016} which also provides finite-sample statistical guarantees and has been successfully deployed to analyze GWAS data \citep{Sesia2019a, sesia2019rejoinder, Sesia631390}, although the connection with causal discovery was not previously developed.

Lastly, we pause to clarify the relationship of this work with Mendelian randomization and fine-mapping. First, Mendelian randomization \citep{katan1986,Smith2003} is a technique that uses a pre-specified genetic variant as an instrumental variable in order to determine the magnitude of the causal effect of an exposure (e.g., smoking) on an outcome (e.g., lung cancer). By contrast, here we are instead interested in identifying which of many genomic regions have a causal effect on an outcome.
Second, methods for {\em fine-mapping} seek to reduce a large region of GWAS detections to a smaller set that is expected to contain the causal variant driving the association \citepeg{Schaid2018}. There, the term ``causal variant'' indicates that one is operating at a high-resolution, whereas the current work is concerned with establishing that associations are truly biological and are not created by confounding.

\section{Causal inference in the trio design}
\subsection{Setting}
\label{sec:setting}
Human cells have 46 DNA strands organized into 23 chromosome pairs; in
each pair, one complete strand is inherited from the mother and one
is from the father. In this work, we consider the case where we
measure {\em single-nucleotide polymorphisms} (SNPs), sites on the
genome where two possible alleles occur in the population,
encoded as 0 or 1. The set of observed alleles for one entire strand
is known as a {\em haplotype}. We consider the case where the
haplotypes of $n$ subjects and their biological parents at $p$ sites
are known, denoted as follows:  
\[
  \begin{array}{rrr}  
\text{subjects' haplotypes} \qquad & (X_1^m, \dots, X_p^m) \in
                                     \{0,1\}^{n \times p}, & (X_1^f, \dots, X_p^f)  \in \{0,1\}^{n \times p}; \\
\text{mothers' haplotypes} \qquad & (M_1^a, \dots, M_p^a) \in
                                    \{0,1\}^{n\times p}, & (M_1^b, \dots, M_p^b) \in \{0,1\}^{n \times p}; \\
\text{fathers' haplotypes} \qquad & (F_1^a, \dots, F_p^a) \in
                                    \{0,1\}^{n \times p}, & (F_1^b, \dots, F_p^b)
                                                 \in \{0,1\}^{n \times
                                                 p}.
                                                 \end{array}
\]
For convenience we denote the matrix of offspring {\em genotypes} as
$X  = X^m + X^f$, the $j$-th column of $X$ as $X_j$, the $i$-th row of
$X$ as $X^\i$ (this is the $i$-th individual), and the set of all ancestral haplotypes as $A = (M^a, M^b, F^a, F^b)$.

Our method takes the haplotypes as given, even though  typically only the genotypes are directly measured in a GWAS study. Haplotypes are then reconstructed algorithmically through {\em phasing} \citep{Browning2011}. While experimental techniques are being developed to directly measure haplotypes, these are not yet widespread. We instead take the phased haplotypes as a reasonable approximation, since phasing is known to be accurate with family data \citep{Browning2011, Marchini2006, OConnell2014}. In a simulation with a synthetic population with known ground-truth haplotypes, we find that our method performs identically with known haplotypes and computationally phased haplotypes; see Appendix~\ref{sec:sim_phasing_robustness}. 

Crucially, the distribution of the offspring genotypes $X$ conditional
on the parental haplotypes $A$ is known
\citep{haldane1919combination}.  Informally, the model for a single
offspring is this: for the haplotype $X^m$ inherited from the mother,
the SNP $X^{m}_{j}$ is inherited either from $M^{a}_j$ or from
$M^{b}_j$, with equal probability. Furthermore, long continuous blocks
of $X^{m}$ are jointly inherited either from $M^{a}$ or $M^{b}$, with
occasional switches at {\em recombination} sites; see Figure
\ref{fig:recom-viz} for an illustration and
Section~\ref{sec:recomb_model} for a formal description of this
process. Throughout, we will leverage our knowledge of the
recombination mechanism in order to carry out hypothesis tests.

\begin{figure}
\begin{center}
\includegraphics[page = 1, width = 4in, trim = 0 20cm 0 0]{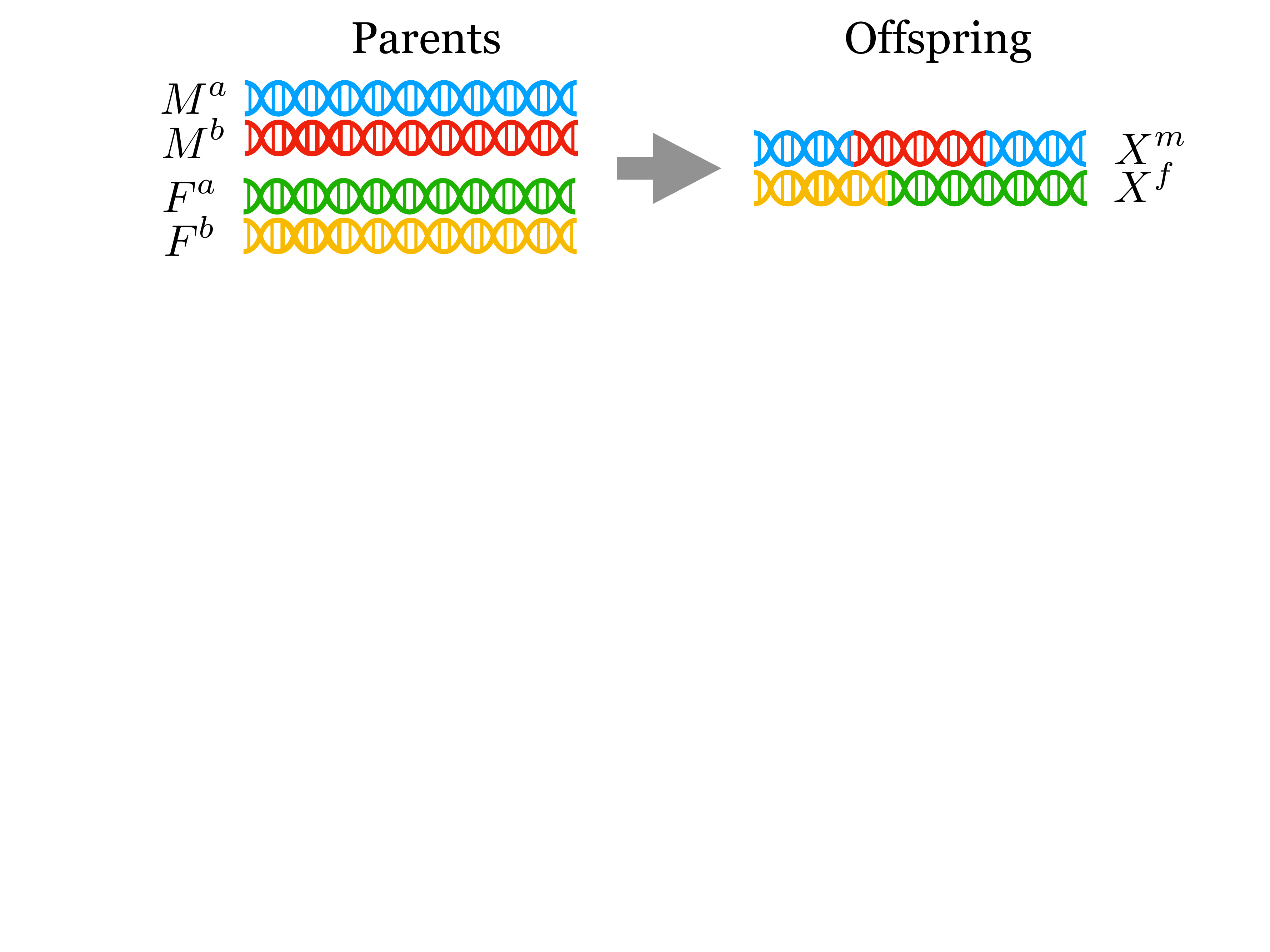}
\end{center}
\caption{A visualization of the process of recombination on a single chromosome.}
\label{fig:recom-viz}
\end{figure}

\subsection{Establishing causality}
\label{sec:causality}

% \begin{tcolorbox}
% {\bf Conditional on the parental haplotypes, we have a randomized experiment}.
% \end{tcolorbox}

We now demonstrate that it is possible to draw causal inference from trio data by formulating the inheritance process as a high-dimensional randomized experiment. The main idea is to condition on the parental haplotypes: once these are fixed the remaining randomness in meiosis is due to cell-level biological processes that are independent of possible confounders, so the resulting inferences are immune to these factors.  

We begin with a concrete example of confounding in GWAS. Suppose we
have a genetic study involving two populations, such as Europeans and
East Asians, and we wish to study whether a SNP $X_j$ affects the cholesterol level $Y$. Next, suppose that the distribution of cholesterol levels differs in the two populations, and furthermore that the distribution of SNP $X_j$ also differs. As a result, there is a valid statistical association between $X_j$ and $Y$, but this statistical association may or may not represent a causal effect. That is, if we manipulated the SNP $X_j$, it may or may not change the cholesterol levels of the subjects. The association could instead be the result a confounder; for example, if the consumption of dairy products leads to higher cholesterol and Europeans consume more dairy products on average, then the SNP $X_j$ will be associated with $Y$, even when it has no causal effect. Since randomized experiments only detect causal effects, to circumvent the above problem we could, in principle, flip a fair coin, set the value of $X_j$ accordingly, and then check for an association with $Y$. While we of course do not carry out such experiments on people, we can exploit a similar experiment occurring in nature.

Formally, let $B \mid C$ denote the distribution of a random variable $B$ given the observed value of a random variable $C$, and let $B \indep C \mid D$ denote that $B$ is conditionally independent of $C$ given $D$. Consider testing the null hypothesis that a SNP $X_j$ is independent of the response $Y$ after conditioning on the parental haplotypes:
\begin{equation}
H_0 : X_j \indep Y \mid A.
\label{eq:cond_null}
\end{equation}
Next, let $Z$ be a potential confounder, such as dairy consumption above. The critical observation is that for essentially all possible confounders of concern in genetic studies, the distribution of the genotype at site $j$ given the parental haplotypes does not change given knowledge of $Z$, since the randomness in inheritance is a result only of random biological processes independent of $Z$. To make this precise, we define the following set of possible confounders.
\begin{definition}[External confounder]
\label{def:ext_confounder}
We say that a random variable $Z$ is an {\em external confounder} if the distribution of the offspring's haplotypes given the parental haplotypes does not change given knowledge $Z$:
\begin{equation}
\label{eq:external_confounder}
X \mid (A, Z = z) \quad \eqd  \quad X \mid (A, Z = z') \quad \text{ for any } z \text{ and } z'.
\end{equation}
\end{definition}
The relation in \eqref{eq:external_confounder} is true for the offspring's dairy consumption, for example, as well as for all environmental conditions occurring after conception. This relationship implies that $Z$ is independent of the offspring's SNP $j$ given the parental haplotypes:
\begin{equation*}
Z \indep X_j \mid A.
\end{equation*}
This then implies that if there is an association between $Y$ and $X_j$ after conditioning on the parental haplotypes, then the association is not due to the confounder $Z$:
\begin{equation*}
Y \notindep X_j \mid A  \implies Y \notindep X_j \mid (A, Z).
\end{equation*}
Returning to the language of hypothesis testing, this proves that if we test the null hypothesis in \eqref{eq:cond_null}, we automatically account for the random variable $Z$. We record this fact formally:
\begin{theorem}[Conditioning on parents accounts for external confounders]
\label{thm:causality}
Let $Z$ be an external confounder, i.e., $Z$ satisfies the relation in
\eqref{eq:external_confounder}. Then, any valid test of the null
hypothesis in \eqref{eq:cond_null} is also a valid test of the stronger null hypothesis 
\begin{equation}
\label{eq:dtt_conf_null}
H'_0 : Y \indep X_j \mid (A, Z)
\end{equation}
that accounts for the confounder $Z$.
\end{theorem}

In words, if we test the hypothesis in \eqref{eq:cond_null}, which is possible based on observed trio data, then we have perfectly adjusted for the confounder $Z$, {\em even if it is not specified or measured in the data}. Thus, if we reject the null in \eqref{eq:cond_null}, it cannot be the case that $X_j$ and $Y$ are dependent due to an external confounder $Z$. We emphasize that rejecting a test of the null hypothesis in \eqref{eq:cond_null} does not yet imply that $X_j$ itself is the causal SNP, since the dependence between $X_j$ and $Y$ may be the result of a causal SNP elsewhere on the chromosome, but it does imply that there is an association on the chromosome that is not the result of external confounding. We will show how to explicitly localize the causal regions later in Section \ref{sec:localization}.

\subsection{Connection to structural equation modeling}

Thus far, we have made statements about conditional independence
relationships without explicitly using existing formalisms for causal
inference. Here, we formulate our results as a structural equation
model to make the connection with the existing literature explicit,
and we similarly formulate our results in the potential outcomes
framework in Appendix~\ref{app:potential_outcomes}.

Consider a structural equation model involving the variables $A, X, Y$ and the external confounder $Z$. For a response $Y$, we assume that $X$ can only cause $Y$ and not the reverse, which is reasonable because a subject's genotype is fixed after conception. We further know that the parental haplotypes $A$ cause $X$ and not the reverse. We also assume that $Z$ causes $Y$ since the reverse case does not result in confounding. Finally, by definition the external confounder $Z$ is conditionally independent of $X$ given $A$, which implies that there is no causal effect from $X$ to $Z$. The corresponding structural equation model is
\begin{align*}
(A, Z) = f_{AZ}(N_{AZ}), \qquad
X = f_X(A, N_X), \qquad
Y = f_Y(X, Z, N_Y), 
\end{align*}
where $f_{AZ}, f_X$, and $f_Y$ are fixed functions and $N_{AZ}, N_X$ and $N_Y$ are independent uniform $[0,1]$ random variables; see Figure \ref{fig:causal_diagrams} for a graphical representation. Within this model, rejecting the hypothesis in \eqref{eq:cond_null} implies that there is a causal effect from $X$ to $Y$ (although not necessarily from the specific SNP $X_j$ to $Y$). In this sense, testing the hypothesis in \eqref{eq:cond_null} is a formal causal inference. Crucially, our proposed method will test this hypothesis without the need to specify or restrict $f_{AZ}$ or $f_Y$.
 
\begin{figure}
	\centering
	\begin{subfigure}[t]{.3\textwidth}
		\centering
        \includegraphics[page=2, width = 1in, trim = 0 500 800 0, clip]{screening_simple.pdf}
        \caption{}
        \label{fig:confounder_dag}
    \end{subfigure}
    ~
    \begin{subfigure}[t]{.3\textwidth}
    	\centering
        \includegraphics[page=1, width = 1in, trim = 0 500 800 0, clip]{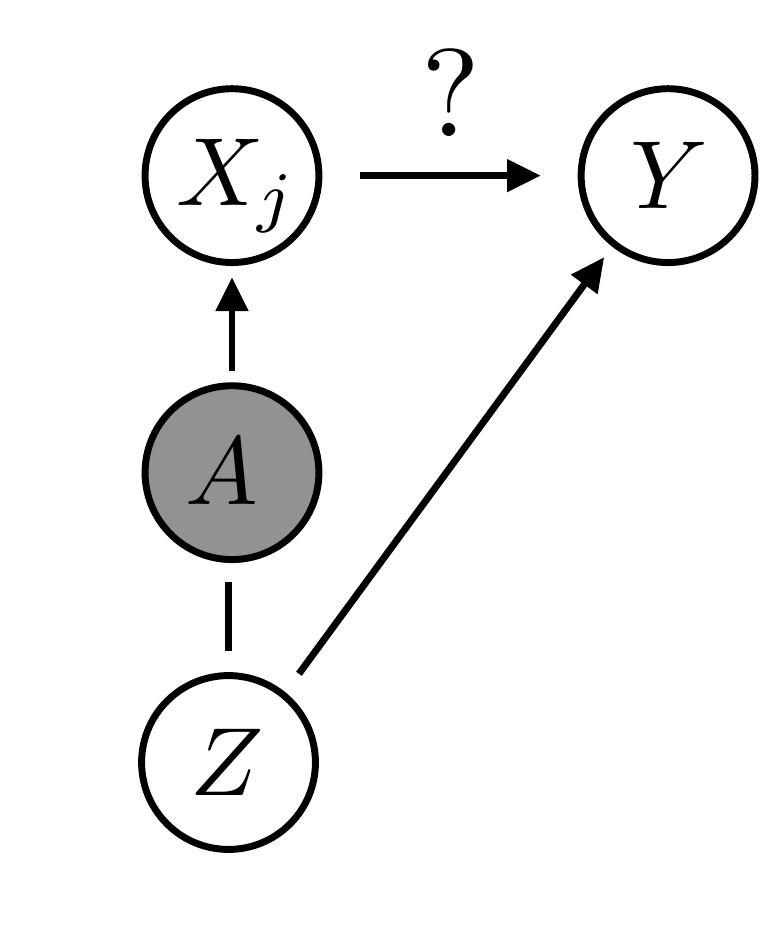}
        \caption{}
        \label{fig:ancestry_dag}
    \end{subfigure}
    
    \caption{A graphical depiction of the causal argument in Section \ref{sec:causality}. Panel (a) shows that the random variable $Z$ can create an association between $X_j$ and $Y$, even if there is no causal effect. Panel (b) shows that conditional on the parental haplotypes $A$, the external confounder $Z$ is independent of the offspring's genotype $X_j$. As a result, $Z$ cannot be responsible for the remaining association between the genotype $X_j$ and the trait $Y$.}
    \label{fig:causal_diagrams}
\end{figure}

\subsection{Discussion of possible confounders}
Virtually all confounders of concern in genetic studies do not affect the transmission of the genetic information from parents to offspring, and are thus external confounders which are correctly accounted for in the trio design by Theorem~\ref{thm:causality}. We list the most important examples below.
\begin{itemize}
\item {\bf Environmental conditions after conception} The mechanism for producing $X$ from $A$ is unaffected by anything occurring after conception.
\item {\bf Population structure, ethnic composition, and geographic location} The mechanism for producing $X$ from $A$ does not change with subpopulation information, ethnicity, or geographic location.
\item {\bf Cryptic relatedness} The presence of distantly related individuals in a sample does not change the distribution of $X$ given $A$, even if this relatedness is unknown and unspecified.
\item {\bf Family effects, altruistic genes} Information about the quality of the environment caused by parental behavior does not impact the distribution of $X$ given $A$.
\item {\bf Assortive mating} Tests of \eqref{eq:cond_null} condition on the observed mating pattern, making them immune to this form of confounding.
\end{itemize}

By contrast, the following are not external confounders:
\begin{itemize}
\item {\bf Germline mutations} A few environmental factors of the parents can affect
  the inheritance process, such as the exposure of a parent to
  radiation, which changes the distribution of the offspring by
  increasing the frequency of mutations. While this does affect the
  model for inheritance in principle, we do not expect this to practically invalidate tests of the null in \eqref{eq:cond_null}.
  In any case, this is a narrow set of possible confounders.
\item {\bf Unmeasured SNPs} In typical studies, only a subset of SNPs are sequenced. Knowledge of a subject's unmeasured SNPs gives additional information about the distribution of $X$ given $A$.
Later in this work, we will show how to localize the unmeasured causal SNPs.
\end{itemize}

Soon, in Section~\ref{sec:global_dtt}, we develop a family of tests of the null in \eqref{eq:cond_null}, but we now pause for two additional remarks about our inferences. First, the reader may wonder how mutations that occur later in a subject's life, such as mutations caused by sun or radiation exposure, fit with the above exposition. Such mutations originate in a single cell, and only descendant cells carry the mutation. As a result, any mutations not present at the earliest stages of development will only be present in a very small minority of a subject's cells. Genotyping techniques use the consensus of measurements from a large number of cells, and so they essentially measure the genotype at conception.

Second, there is potential for selection bias in all genetic studies,
since some individuals are more likely to be included in a sample than
others. Tests of the null in \eqref{eq:cond_null} remove the effect of selection bias, except for
that directly caused by the offspring's genetic composition. Any
potential selection bias due to external confounders, such as
geographic location, is automatically accounted for by
Theorem~\ref{thm:causality}. However, if a SNP $X_j$ causally
influences the probability of inclusion in a study, then such a SNP is
not null according to our null hypothesis in $\eqref{eq:cond_null}$,
so it may be detected by such a test.

\subsection{The randomness in inheritance}
\label{sec:recomb_model}
In preparation for a test of hypothesis in \eqref{eq:cond_null}, we now formally describe the distribution of offspring's genotype given the parental haplotypes. The biological mechanism governing inheritance is well-understood, and this will serve as the backbone of our statistical inference. In particular, the process by which a subject's two haplotypes arise from the parental haplotypes was formalized by Haldane as a hidden Markov model (HMM) \citep{haldane1919combination}. Without loss of generality, we describe the model for a single observation on one chromosome; in the general case, each chromosome of each observation is an independent instance of this model. For concreteness, we focus on $X^m$.

\paragraph{Latent Markov chain} Let the random vector $U^m \in \{a, b\}^{p}$ indicate the following:
\begin{align*}
U^{m}_j &= \begin{cases}
a & \text{ if site $j$ is copied from the mother's `a' haplotype}, \\
b & \text{ if site $j$ is copied from the mother's `b' haplotype}.
\end{cases} \\
\end{align*}
Our model is that $U^m$ is distributed as a Markov chain, where
\begin{align*}
P(U^{m}_1 = a) &= \frac{1}{2},
\end{align*}
and
\begin{align*}
P(U^{m}_j = u^{m}_{j-1} \mid U^{m}_{1:(j-1)} = u^m_{1:(j-1)}) &= P(\text{even} \overbrace{\text{ \# of recombinations between $j-1$ and $j$}}^{\text{modeled as Poisson with mean $d_j$}}) \\
&= \frac{1}{2} (1 + e^{- 2 d_j}).
\end{align*}
Here, $d_j$ is the {\em genetic distance} between SNPs $j-1$ and $j$, which is fixed and known.
Note that the genetic distance is not always proportional to the physical distance due to recombination hotspots: regions that have more frequent recombination events \citep{altshuler2005haplotype, Bherer2017}.

\paragraph{Emission distribution} 
Conditional on $U^m$, each $X_j^m$ is independently sampled from
$$P(X^{m}_j = M^{(u^m_j)}_j \mid U^{m}_j = u^{m}_j) = 1 - \epsilon.$$
Here, $\epsilon$ is the probability of a de novo mutation, which for humans is about $1\cdot10^{-8}$ \citep{Acuna-Hidalgo2016}. The analogous HMM describes the distribution of $X^f$ given $F^a$ and $F^b$, which is taken to be independent of $X^m$ given $A$.

\section{Testing the global causal null}
\label{sec:full_chrom_testing}

The remainder of this work develops randomization tests to address two questions of interest.
First, we wish to determine whether a trait has any genetic basis, that is, to test the global null. Second, we wish to find regions of the genome that contain causal variants. 
To clarify the relationship between these questions and our choice of null hypothesis, consider $H_0$ in \eqref{eq:cond_null}. This hypothesis probes the association between the trait of interest and SNP $X_j$, without considering the role of other genetic variants. It is non-null either when $X_j$ is directly related to the trait or when $X_j$ is merely associated with a causal genetic variant. Therefore, to establish that $X_j$ itself is the causal variant, we additionally need to ensure that the association between $X_j$ and $Y$ is not due to any other genetic variables. One removes these possible confounders by conditioning on them, testing instead hypotheses that check the dependence of $X_j$ and $Y$ after accounting for the rest of the genome \citep{candes2016,Sesia631390}. We will describe how to do this in detail soon, but we start by noting that once we condition on the parental haplotypes, the SNPs on different chromosomes are independent. As a result, in the trio design even global null tests can localize signals to a single chromosome.  

\subsection{The full-chromosome Digital Twin Test}
\label{sec:global_dtt}
We begin by testing the a full-chromosome global null.
Let $\C \subset \{1,\dots,p\}$ be the set of SNPs on the given chromosome. We now consider the hypothesis
\begin{equation}
\label{eq:composite_cond_null}
H_0 : X_\C \indep Y \mid A,
\end{equation}
where $X_\C$ denotes $(X_j)_{j \in \C}$. 
The {\em Digital Twin Test} of the null hypothesis in \eqref{eq:composite_cond_null}---presented in Algorithm~\ref{alg:global_dtt}---compares an observed offspring to synthetic offspring from the same parents (the ``digital twins''; see Figure \ref{fig:digital_twin_viz}) using an arbitrary measure of feature importance or predictive accuracy.\\
\begin{algorithm}[H]
\SetAlgoLined
\DontPrintSemicolon
--- Compute the test statistic on the true data: $$t^* = T((X_{-\C}, X_\C), Y).\;$$
\For{$k = 1, \dots, K$}{
	--- Sample the digital twins $\wtilde{X}^m_\C$ and $\wtilde{X}^f_\C$ from the distribution from Section \ref{sec:recomb_model}, independent of $Y$ (see Appendix~\ref{sec:sampling_digital_twins} for an explicit sampler). \vspace{.05in}\;
	--- Compute the test statistic using the digital twins: $$t_k = T((X_{-\C}, \Xk^m_\C + \Xk^f_\C), Y).\;$$
}
--- Compute the quantile of the true statistic $t^*$ among the digital twin statistics $t_1, \dots, t_K$:
$$v = \frac{1 + \#\{k: t^* \le t_k\}}{ K+ 1}.\;$$

\KwResult{A p-value $v$ for the conditional hypothesis in \eqref{eq:composite_cond_null}.}
\caption{The Digital Twin Test}
\label{alg:global_dtt}
\end{algorithm}
In Algorithm~\ref{alg:global_dtt}, $X_{-\C}$ denotes
$(X_j)_{j \notin \C}$, and $T(\cdot)$ can be any statistic. This
algorithm returns a p-value, and the corresponding level $\alpha$
hypothesis test rejects when this p-value is less than
$\alpha$. We record formally that this test accounts for external confounders; see Appendix~\ref{app:proofs} for all proofs.

\begin{figure}
\begin{center}
\includegraphics[page = 2, width = 4in, trim = 0 20cm 0 0]{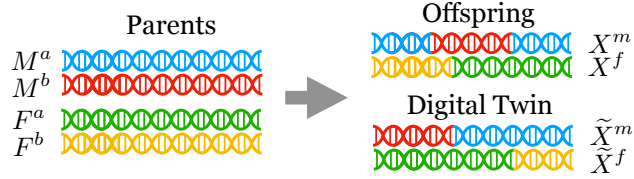}
\end{center}
\caption{A visualization of a digital twin.}
\label{fig:digital_twin_viz}
\end{figure}

\begin{proposition}[The Digital Twin Test is immune to external confounders]
\label{prop:chromosome_dtt_causality}
The Digital Twin Test is a valid test of the hypothesis 
\begin{equation}
H_0^\C : Y \indep X_\C \mid (A, Z), 
\label{eq:chromosome_conf_null}
\end{equation}
where $Z$ is any (possibly unknown) external confounder. That is, if the null hypothesis in \eqref{eq:chromosome_conf_null} holds, then for all $\alpha \in (0,1)$, the output $v$ of Algorithm~\ref{alg:global_dtt} obeys
\begin{equation*} 
P(v \leq \alpha) \le \alpha.
\end{equation*}
\end{proposition}

The Digital Twin Test is a natural randomization test in a trio design because it replicates the random mechanism generating the data. By contrast, a traditional permutation test can only test for a marginal association between $X_j$ and $Y$, which does not lead to valid causal inferences; see Appendix~\ref{sec:perm_failure}.

Many existing tests fall into this family. Notably, the TDT is a special case of the Digital Twin Test with the test statistic
\begin{equation}
\label{eq:tdt_test_stat}
T^{\text{(TDT)}}(X) = \sum_{i = 1}^n X_j^{\i} \I_{\{Y_i = 1\}},
\end{equation}
where $\I_B$ denotes the indicator of event $B$, except that the
Digital Twin Test uses an exact, finite-sample rejection threshold,
whereas the TDT uses an asymptotic approximation. Similarly, the
quantitative TDTs \citep{Rabinowitz1997} and \citep{Monks2000} are
also special cases of the above procedure. Moreover, the Digital Twin
Test can exploit arbitrary black-box machine learning models, such as
deep neural networks, gradient boosting, random forests, and penalized
regression to form a test statistic $T(\cdot)$ that incorporates
information from multiple sites in a data-driven way; see
\eqref{eq:test_stat} below for a concrete example. More sophisticated
models can explain away more of the variation in the phenotype,
leading to more sensitive tests. Thus we can see that the Digital Twin
Test framework unifies many existing procedures while
incorporating varying disease models, subject matter knowledge,
fitting algorithms, principal component corrections, screening and
replication, and so on, without requiring a new mathematical analysis
for each case.
While well-chosen models will lead to more powerful tests, we
emphasize that the validity of the automatic, finite-sample inference does not depend on the correctness of the chosen model.

\subsection{Incorporating external GWAS data}
The Digital Twin Test can also leverage large external GWAS data sets that do not contain trio observations to increase power. This is important because such data sets are common and will typically have larger sample sizes than trio data sets. To this end, we can use the external GWAS to find a powerful test statistic $T(\cdot)$, as suggested by \citep{Tansey2018}. For example, suppose we fit a penalized linear or logistic regression model on the external GWAS data to obtain an estimated coefficient vector $\hat{\beta}$. Then, on the trio data, we can use the Digital Twin Test with test statistic
\begin{equation}
\label{eq:test_stat}
	T(X, Y) = \begin{cases} 
	-\sum_{i=1}^n (\hat{\beta}^\top X^\i - Y_i)^2 & \text{ for real-valued $Y$}, \\
	-\sum_{i=1}^n - Y_i \log(\frac{e^{\hat{\beta}^\top X^\i}}{1 + e^{\hat{\beta}^\top X^\i}}) - (1 - Y_i) \log(\frac{1}{1 + e^{\hat{\beta}^\top X^\i}})& \text{ for binary $Y$}
	\end{cases}
\end{equation}
(the negative squared loss and the negative
logistic loss).
In words, the Digital Twin Test with this test statistics is asking, ``are the residuals smaller when I use the real genotypes to predict the response, compared to when I use the digital twin genotypes to predict the response?'' If the residuals are systematically smaller, it must be because of a causal effect, and the Digital Twin Test rejects the null hypothesis. 
We will explore this approach in simulation in Section~\ref{sec:simulations}. 

\subsection{Parent-offspring duos and other pedigrees}
The Digital Twin Test can also be applied to offspring for whom only one parent is genotyped, with a small adjustment. The modification is simple: whenever a parent is unknown the algorithm fixes the offspring's haplotype from that parent. For example, if $F^a$ and $F^b$ are unknown then in Algorithm~\ref{alg:global_dtt} we set $\wtilde{X}_g^f = X_g^f$ in each iteration of the loop. An analogous version of Theorem~\ref{thm:causality} continues to hold in this setting, and so we still detect only causal regions.\footnote{We simply condition on $X^f$ rather than $F^a$ and $F^b$ for each unit where the father's haplotypes are unknown, and so on.}

Furthermore, the Digital Twin Test can be applied to data with a variety of pedigree structures. One can select any set of duos or trios from the pedigree, with the restriction that no offspring in a duo or trio is an ancestor of any other offspring in a duo or trio. Even though individuals may be related, the locations of their recombination sites are independent. Moreover, only the parent-offspring relations need to be known, not the full pedigree. 

\section{Localizing causal regions}
\label{sec:localization}

\subsection{Linkage disequilibrium in the trio design}
\label{sec:trio_ld}
SNPs on the same chromosome are dependent, a phenomenon referred to as {\em linkage disequilibrium} (LD), so the Digital Twin Test and TDT can only rigorously test the full-chromosome causal null hypotheses discussed in Section~\ref{sec:full_chrom_testing}. Before presenting in this section a version of the Digital Twin Test to localize causal SNPs, we highlight the current limitation with a practical example.

\begin{figure}
\centering
\includegraphics[width = 6.5in]{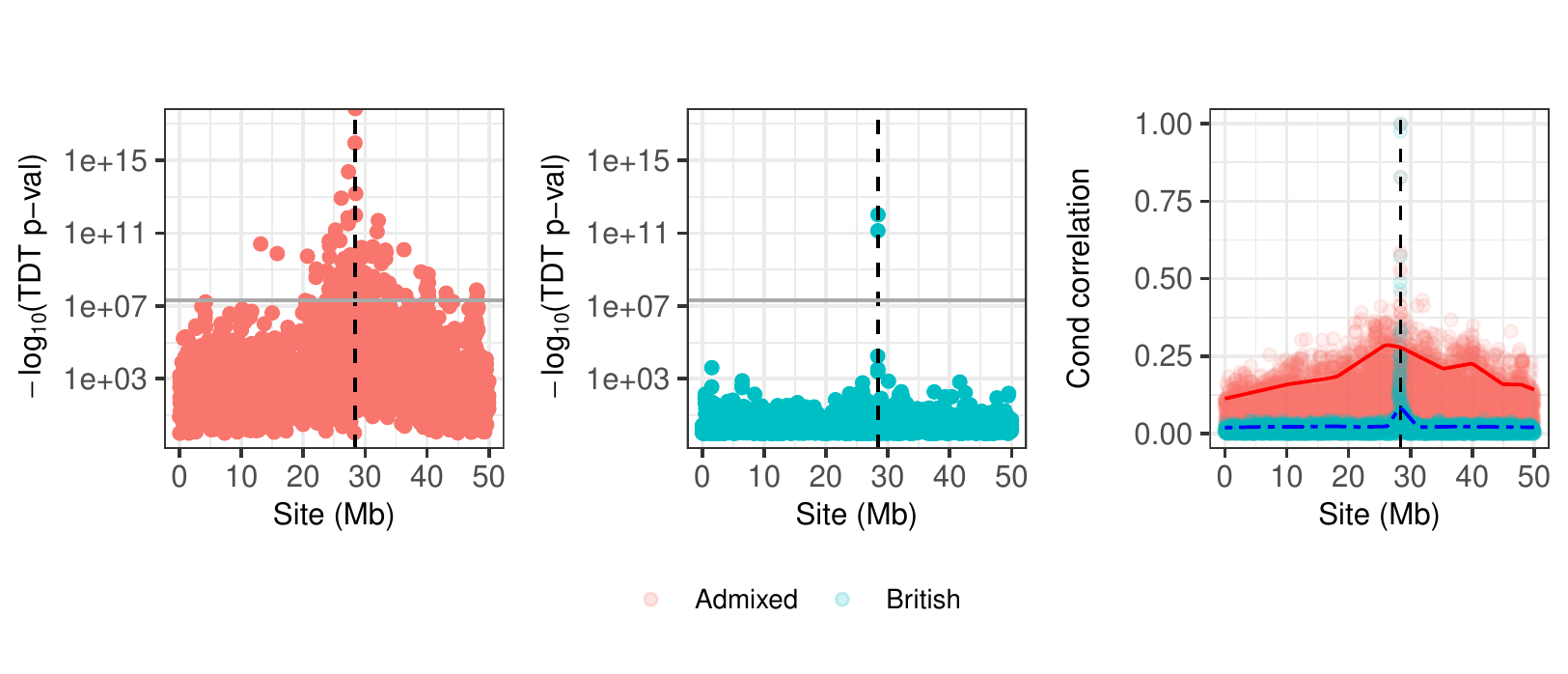}
\caption{Results of the TDT in two populations.  Left and center: Manhattan plots on chromosome 22, which contains the one true causal SNP, indicated with a dashed vertical line. The genome-wide significance threshold is shown with a gray horizontal line. The left panel shows an admixed population, whereas the center panel shows a British population. Right: a plot of the absolute correlations between the causal SNP and the other SNPs, conditional on the parental haplotypes. The red solid and blue dot-dashed curves indicate a smoothed 90\% quantile of the absolute correlation with the causal SNP across the chromosome, for the admixed and British populations, respectively.}
\label{fig:admixed_manhattan}
\end{figure}

We create a synthetic population of 2,500 second-generation admixed
individuals whose parents are the children of one ethnically British
individual and one ethnically African individual. The haplotypes of
the parents are real in the sense that they are phased haplotypes from
the UK Biobank data set \citep{bycroft2018}. For simplicity, we create
a binary synthetic response $Y$ from a logistic regression model with
a single causal SNP. We choose a signal strength such that
that the heritability of the trait is 18\% and an intercept such that
20\% of the observations have $Y=1$. 
The causal SNP is chosen at a site with a large difference in allele
frequency between the British and African populations. Then, we carry
out the TDT at each site and report the p-values in a Manhattan plot
in Figure \ref{fig:admixed_manhattan}. Note that even if we demand
that the p-values are smaller than the genome-wide significance
threshold of $5\cdot10^{-8}$, the TDT reports discoveries all across
the chromosome. We compare this to an identical simulation comprised
only of British individuals in the center panel of Figure
\ref{fig:admixed_manhattan}. Here, the TDT only reports discoveries near the true causal SNP. 

The TDT behaves differently in these two populations because of the
different correlation structure after conditioning on the parental
haplotypes. In the admixed population, there are large correlations
between sites far away, but not in the British population; see the
right panel of Figure \ref{fig:admixed_manhattan}. Because of the large LD in the admixed population, testing the full-chromosome causal null in \eqref{eq:cond_null} cannot give reliable information about the location of the causal SNP---this limitation applies both to the TDT and to the Digital Twin Test from Section~\ref{sec:full_chrom_testing}. This weakness of the TDT has been noticed before in in the admixed setting \citep{wang2016}. That work developed an analytical correction based on population-genetic quantities. Instead, we will next extend the Digital Twin Test to provably localize the causal regions using conditional independence testing.

\subsection{The Local Digital Twin Test}
\label{sec:cond_dtt}

Since the location of the causal SNPs cannot be reliably found from the TDT or any other full-chromosome test, we next introduce the {\em Local Digital Twin Test} to identify distinct causal regions. Our method partitions the genome into disjoint regions and then constructs a p-values for the hypothesis that the region contains no causal SNPs. Because we investigate many regions simultaneously, it is critical to leverage modern methods for multiple testing, such as the Benjamini--Hochberg procedure \citep{Benjamini1995}, in order to have higher power than conservative methods like Bonferroni. To this end, we will also develop a technical modification of the Local Digital Twin Test that yields {\em independent} p-values, which can be used as inputs to more powerful multiple testing procedures.

Formally, let $G$ be a partition of $\{1,\dots,p\}$. For each group of SNPs $g \in G$, we will consider the hypothesis
\begin{equation}
H_0^g : Y \indep X_g \mid (X^m_{-g}, X^f_{-g}, A).
\label{eq:group_null}
\end{equation}
In words, this is the hypothesis that knowing the SNPs in group $g$ is
not informative about the response once we know the remaining SNPs and
the parental haplotypes. This ensures that any detections are the results of causal SNPs in the region $g$ rather than elsewhere on the chromosome.
In what follows, we will consider groups of SNPs $g$ that form a continuous block of the chromosome, but this is not required. Notice that larger group sizes correspond to weaker statistical statements and hence the corresponding tests have higher power \citepeg{Dai2016, Sesia631390}. % Because we are conditioning on the ancestors $A$, such tests will be immune to external confounders, as before.

The Local Digital Twin Test---presented in Algorithm~\ref{alg:ldtt}---tests the null hypothesis in \eqref{eq:group_null} by again creating synthetic offspring from a subject's parents (the ``local digital twins''). Unlike before, the synthetic offspring are now additionally constrained to match outside the region $g$. That is, they are sampled from the distribution of
\begin{equation}
(X^m_g, X^f_g) \mid (X^m_{-g}, X^f_{-g}, A).
\label{eq:ldtt_imputation_dist}
\end{equation}
See Figure~\ref{fig:local_digital_twin_viz} for an illustration.
%Conditioning on the SNPs outside of region $g$ ensures that any effect that is detected is the result of SNPs in the group $g$ and is not the result of correlations between SNPs in group $g$ and others outside the group.

\begin{algorithm}
\SetAlgoLined
\DontPrintSemicolon

--- Compute the test statistic on the true data: $$t^* = T((X_{-g}, X_g), Y).\;$$
\For{$k = 1, \dots, K$}{
	--- Sample the digital twins $(\wtilde{X}_g^m, \wtilde{X}_f^m)$ from the distribution \eqref{eq:ldtt_imputation_dist}, independent from $(X^m_g,X^f_g)$ and $Y$ (see Appendix~\ref{sec:sampling_digital_twins} for an explicit sampler).\;
	--- Compute the test statistic using the digital twins: $$t_k = T((X_{-g}, \Xk^m_g + \Xk^f_g), Y).\;$$
}
--- Compute the quantile of the true statistic $t^*$ among the digital twin statistics $t_1, \dots, t_K$:
$$v = \frac{1 + \#\{k: t^* \le t_k\}}{ K+ 1}.\;$$

\KwResult{A p-value $v$ for the conditional hypothesis in \eqref{eq:cond_null}.}
\caption{The Local Digital Twin Test}
\label{alg:ldtt}
\end{algorithm}
This test both accounts for external confounders and localizes the causal regions, as stated next.
\begin{figure}
\begin{center}
\includegraphics[page = 3, width = 4in, trim = 0 20cm 0 0]{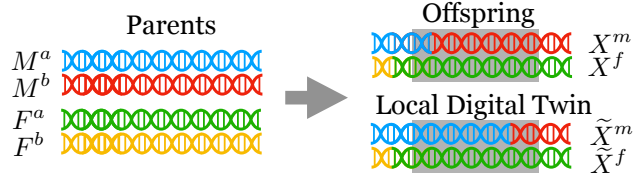}
\end{center}
\caption{A visualization of a local digital twin. The gray shaded region represents the group $g$; the local digital twin always matches the true offspring outside this region.}
\label{fig:local_digital_twin_viz}
\end{figure}
\begin{proposition}[The Local Digital Twin Test is immune to external confounders]
The Local Digital Twin Test is a valid test of the hypothesis 
\begin{equation}
H_0^g : Y \indep X_g \mid (X^m_{-g}, X^f_{-g}, A, Z), 
\label{eq:group_conf_null}
\end{equation}
where $Z$ is any (possibly unknown) external confounder. That is,
under \eqref{eq:group_conf_null}, for all $\alpha \in (0,1)$, the
output $v$ of Algorithm~\ref{alg:ldtt} satisfies
\begin{equation*} 
P(v \leq \alpha) \le \alpha.
\end{equation*}
\label{prop:ldtt_causal_validity}
\end{proposition}

% When considering many different groups, independent p-values enable the analyst to be less conservative than the usual Bonferroni correction by leveraging sophisticated multiple testing methods.
The Local Digital Twin Test p-values above are not guaranteed to be
independent for null groups, so we present a modification that gives
provably independent p-values in Appendix~\ref{app:independent_p}; see
Algorithm~\ref{alg:ldtt-independent} therein. The main idea is to
isolate the randomness used to test each region. The details require substantial additional notation, however, so they are deferred to the Appendix.

\begin{theorem}[Independence of null  p-values]
Algorithm~\ref{alg:ldtt-independent} (defined in Appendix~\ref{app:independent_p}) produces p-values $v_g$ satisfying the following properties:
\begin{enumerate}
\item For groups $g$ satisfying the Local Digital Twin Test null hypothesis in \eqref{eq:group_conf_null}, the distribution of $v_g$ stochastically dominates the uniform distribution, i.e., $v_g$ is a valid p-value.
\item The p-values for null groups are jointly independent.
\end{enumerate}
\label{thm:independent_pvals}
\end{theorem}

As before, any score function $T(\cdot)$ leads to valid inference, but some functions will lead to higher power. One option requiring only trio data is the statistic in \eqref{eq:tdt_test_stat}. This gives a TDT-like test, with additionally localization guarantees at the price of lower power. A better option is to incorporate external GWAS data with the previously-discussed test statistic in \eqref{eq:test_stat}.

\subsection{Weighted and ordered testing}

To further increase power, the external GWAS data should be used to
prioritize the most promising regions of the genome for confirmatory
testing with the Local Digital Twin Test; \citep{roeder2009} gives a general discussion of incorporating weighted testing in GWAS and \citep{ren2020knockoffs} show how to use information to improve the ordering for knockoff testing. As a concrete example, when using the test statistic in \eqref{eq:test_stat} one could order the hypothesis by decreasing value of 
\begin{equation}
\label{eq:beta_hat_ordering}
w_g = \sum_{j\in g} |\hat{\beta}_j|,
\end{equation} 
and then use the Selective SeqStep procedure \citep{barber2015} or an accumulation test \citep{li2017} to give a final selection set with guaranteed FDR control. We numerically explore this approach in Section~\ref{sec:simulations}.

\subsection{Resolution of the Local Digital Twin Test}
The Local Digital Twin Test will depend primarily on those samples
where a recombination is observed in the region $g$ under study. For a
single subject, consider $\wtilde{X}^m_g$ in
Algorithm~\ref{alg:ldtt}. Suppose that no recombination event occurs
in the region $g$ for the maternal haplotype $X^m$. Then, the
subject's haplotype $X^m$ is copied from the same of the mother's two
haplotypes (for example, $M^a$) at both edges of $g$, so the most
likely sample from distribution in \eqref{eq:ldtt_imputation_dist} is such that the entire region including $g$ matches $M^a$, in which case $\wtilde{X}^m_g = X^m_g = M^a_g$; see $\wtilde{X}^f$ and $X^f$in Figure~\ref{fig:local_digital_twin_viz}. These samples are not informative, since there is no contrast between the true observations and the digital twins to use for inference. On the other hand, if a recombination takes place in the group $g$, then $\wtilde{X}^m_g$ and $X^m_g$ will differ; see $\wtilde{X}^m$ and $X^m$ in Figure~\ref{fig:local_digital_twin_viz}.  Putting this together, only samples that contain a recombination event in the group $g$ will be of use.

This behavior limits the resolution of the Local Digital Twin test. In humans, the probability of a recombination event in a 1 Mb region is about a 1\% \cite{altshuler2005haplotype}. Thus, for a data set of $100,000$ trios, each region of size 1 Mb would only have about $2,000$ samples with recombination events:
\begin{equation*}
\underbrace{100,000}_{\text{sample size}} \times \underbrace{1 \ \text{Mb}}_{\text{size of region}}
\times \underbrace{0.01 \text{ \ recombinations / Mb}}_{\text{recombination rate}} \times \underbrace{2}_{\text{Haplotypes per subject}} = \underbrace{2,000.}_{\text{recombination events}}
\end{equation*}
As a result, the Local Digital Twin Test would essentially compare the $2,000$ subjects with recombination events to their digital twins. The remaining observations could be used for fitting the multivariate model, but would not be used for confirmatory inference. This reduction in effective sample size at fine resolution limits the precision with which the Local Digital Twin Test can make discoveries, and is the price we pay for making stronger causal inferences.
Nonetheless, our goal is not single-SNP precision, but rather is to rigorously establish that there are distinct causal regions across the genome.
%Moreover, for a fixed effect size and power level, the resolution achievable by the Local Digital Twin Test is linear in the number of observations, so because of the exponential growth in sequencing activity \citep{NIHSeequence}, the required group sizes will roughly be cut in half each year in the next decade.

\section{Simulation experiments}
\label{sec:simulations}
In this section, we examine the performance of the Digital Twin Test in semi-synthetic examples, focusing on the binary response case, so that the standard TDT can serve as a benchmark. We form our parent-offspring population by taking real haplotypes from the UK Biobank data set and sample offspring according to the recombination model from Section \ref{sec:recomb_model}. In each experiment, we sample the offspring once and then repeat the generation of the synthetic phenotype multiple times, indicating the standard error with error bars. When presenting the results, we index the signal strength by {\em heritability}: a $[0,1]$-valued scale defined in Appendix~\ref{app:heritability}. An R package implementing the methods below together with notebook tutorials is available at \url{https://github.com/stephenbates19/digitaltwins}.

\subsection{Testing the full-chromosome causal null}
\label{sec:global_null_dtt_sim}
We first examine the ability of the Digital Twin Test to test the full-chromosome causal null. In this simulation, we test only one hypothesis, so we seek to control the usual type-I error rate at the $\alpha = 0.05$ level. We create a synthetic population of $n=2,500$ parent-child trios and generate a binary valued response coming from a sparse logistic regression model
\begin{equation}
\label{eq:logistic_causal_model}
\log \left(\frac{P(Y_i = 1)}{P(Y_i = 0)}\right) = \beta_0 + \beta^\top X^\i, 
\end{equation}
with 10 nonzero entries of $\beta$ of equal strength, chosen uniformly at random.
The intercept $\beta_0$ is chosen so that the fraction of cases is $50 \%, 20 \%,$ or $5 \%$. We use $p = 6,820$ SNPs from chromosome 20, which has width 63 Mb.
We emphasize that the above gives a well-defined structural equation model on $(X, Y)$, and the nonzero entries of the coefficient vector $\beta$ correspond to the SNPs that have a causal effect.\footnote{The reader may wonder about the identifiability of this model. Note that both $X$ and $Y$ are random variables, so provided that in the distribution of $X^{\i, m} + X^{\i, f}$ is not contained in a subspace of rank less than $p$, then this model is identified.} 

In order to use the Digital Twin Test for confirmatory analysis, we
take an external GWAS data set that contains $7,500$ other non-trio
observations from the UK Biobank. With these data, we fit a
$\ell_1$-penalized logistic regression model (with regularization
parameter chosen by cross-validation) to obtain a predictive model for
the trait. We denote the resulting coefficient estimate
$\hat{\beta}$. Then, we apply the Digital Twin Test on the trio data
with feature importance statistic in \eqref{eq:test_stat} to produce a
single p-value, rejecting when it falls below $\alpha = 0.05$.

\begin{figure}
\includegraphics[width = 6.5in]{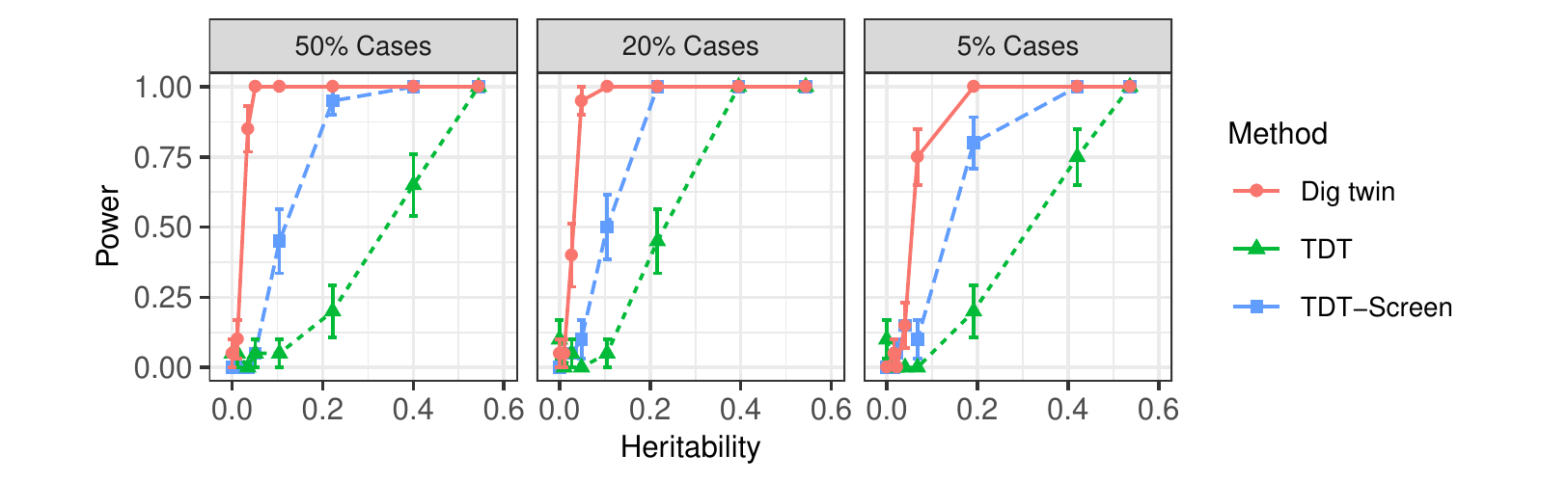}
\caption{Power of the Digital Twin Test compared to TDT benchmarks for testing the full-chromosome causal null. }
\label{fig:full_chrom_power}
\end{figure}

We take the TDT as a natural benchmark, interpreting its output in two
alternative ways. First, we take the minimum p-value after applying the
TDT at every SNP and then Bonferroni correct it (this does not use the
non-trio GWAS data). Second, we compute the Bonferroni-adjusted minimum
p-value only on the coordinates with nonzero coefficients in the lasso
fit $\hat{\beta}$ on the external GWAS data. Because $\hat{\beta}$ is
sparse, this method has a less severe Bonferroni correction and may be
more powerful than the other TDT procedure.

Since all three methods are valid tests of the null hypothesis that there is no causal SNP on the chromosome, we directly compare their power in Figure \ref{fig:full_chrom_power}, using 20 independent realizations for each data point. We find that the Digital Twin Test has higher power than the TDT, even when the latter attempts to leverage the external GWAS as a screening step. The leftmost point in each panel is the null case with zero heritability; the empirical error of the Digital Twin Test does not exceed the nominal level of $\alpha = 0.05$ in any of the three cases.

\subsection{Localization}
\label{sec:sims_large_scale}

We now examine the ability of the Local Digital Twin Test to identify causal regions. Here, we use $p = 591,513$ SNPs on chromosomes $1-22$, split into $532$ pre-determined groups of size approximately 5 Mb. 
The response is again generated from the logistic regression model in \eqref{eq:logistic_causal_model}, and the number of non-zero coefficients in the true causal model is varied as a control parameter.
We consider a sample of $n = 10,000$ trios with an external GWAS of
size $50,000$ used to fit a logistic regression model $\hat{\beta}$ as in the previous section. The fitted coefficients $\hat{\beta}$ are used to form the test statistic in \eqref{eq:test_stat}.
Here, we take the nominal level for the FDR to be $\alpha = 0.2$. Each experiment is repeated $10$ times. Additional technical details about these simulations can be found in Appendix~\ref{sec:app_sims_large_scale}.

We compare the following procedures:
\begin{enumerate}
\item {\bf Digital Twin Test---Accum}  We apply the Local Digital Twin Test at each of the groups to obtain one p-value per group. We also use the external GWAS data to order the regions from most to least promising as in \eqref{eq:beta_hat_ordering}, and use an accumulation test \citep{li2017} to produce a final set of discoveries. This method is guaranteed to control the FDR.\footnote{Strictly speaking, this procedure controls a modified version of the FDR (see \citep{li2017}), but the difference will be unimportant in settings with a large number of discoveries. This is a property of the accumulation test, not the Local Digital Twin Test, and other procedures can be used for standard FDR control.}

\item {\bf TDT---Sreen---BH}: For each group, we apply the TDT to the SNPs with nonzero entries of $\hat{\beta}$, the model fit on the external GWAS. Then, we report the minimum p-value after adjusting it with Bonferroni. Lastly, we apply the Benjamini--Hochberg procedure to report a set of groups. This method assumes the TDT p-values are valid for the group null hypothesis in \eqref{eq:group_null}, which is not fully correct, so this method does not have formal guarantees.

\item {\bf TDT---BH}: We proceed as above, except that we apply the TDT to all SNPs. This method also incorrectly assumes the TDT p-values are valid for the group null hypothesis in \eqref{eq:group_null}, so it does not have formal guarantees.
\end{enumerate}

\begin{figure}
  \centering
  \includegraphics[width=6.5in]{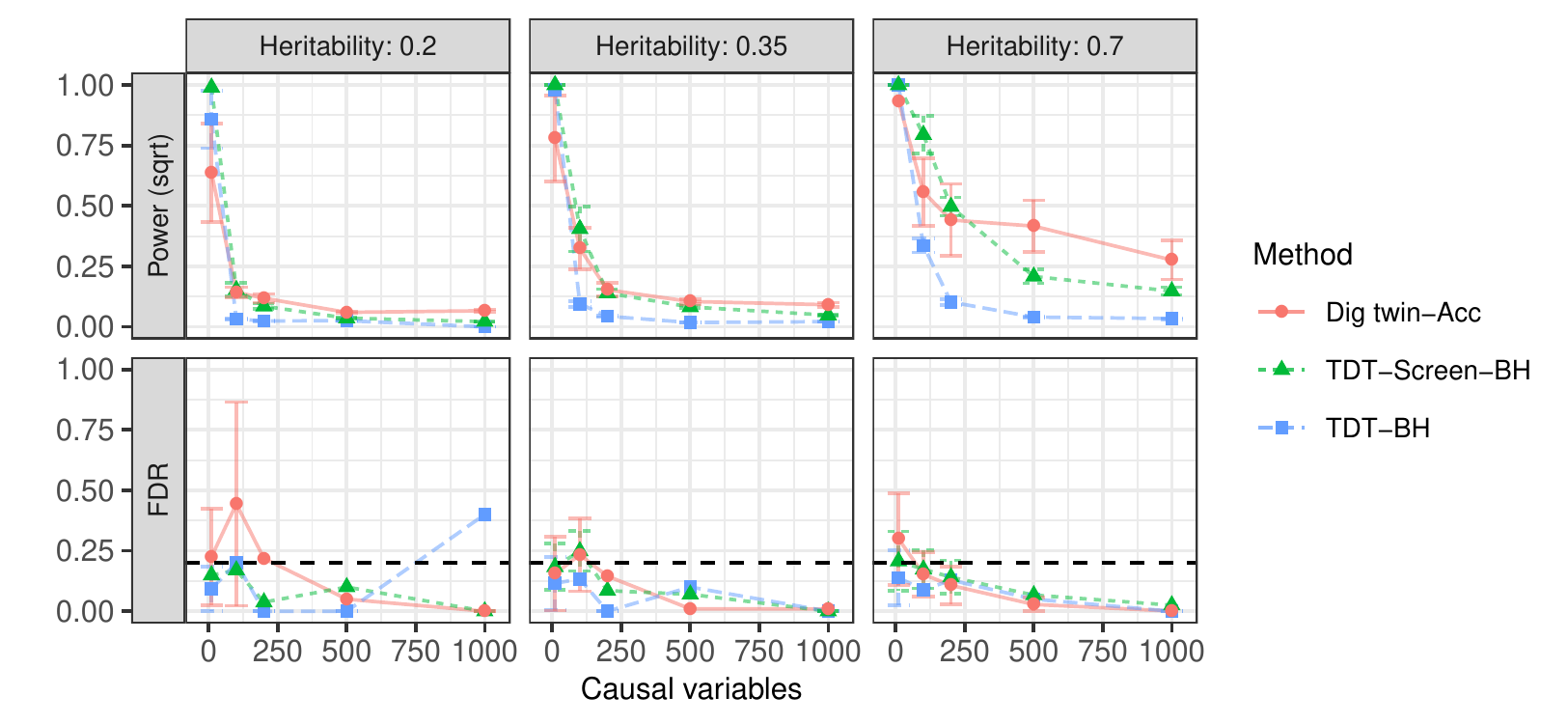}
  \caption{Performance of the Local Digital Twin Test and TDT in the binary-response full-genome simulations from Section~\ref{sec:sims_large_scale}. Here, error bars give one standard deviation, and the dashed horizontal line
  indicates the nominal FDR level}
  \label{fig:exp_large_binary_causal}
\end{figure}

We report the results in Figure~\ref{fig:exp_large_binary_causal}. The
Local Digital Twin Test has generally comparable power to the screened
TDT method with Benjamini--Hochberg, with moderate power improvements
for traits caused by many SNPs. We report on a similar experiment with a continuous
response in Appendix~\ref{app:sim_full_gen_cts}.

Although the TDT benchmarks empirically control the FDR here, we
emphasize that it formally tests the full-chromosome null in 
\eqref{eq:cond_null}, so it is not valid for our goal of localizing the
causal SNPs into the given groups. The group hypothesis in
\eqref{eq:group_null} can only be rigorously tested by a test of the null in
\eqref{eq:cond_null} if the SNPs in different groups are
independent. If the TDT-based group p-values were valid, these two benchmarks would control the FDR. This may be a reasonable approximation within this
experiment, because the groups are wide and the population is
homogeneous, so the LD decays rapidly. However, this assumption fails
in other cases, as we previewed in Section~\ref{sec:trio_ld} and will
return to next.

\subsection{Spurious discoveries with the TDT}
\label{sec:exp_tdt_loc}

We have seen in Section~\ref{sec:trio_ld} an example where the TDT makes false discoveries throughout the chromosome because it does not account for LD. Here, we revisit this example more carefully, demonstrating that the TDT-based benchmarks from Section~\ref{sec:sims_large_scale} can dramatically fail to control type-I errors for the group hypotheses.

\begin{figure}[p]
\includegraphics[height = 2.5in]{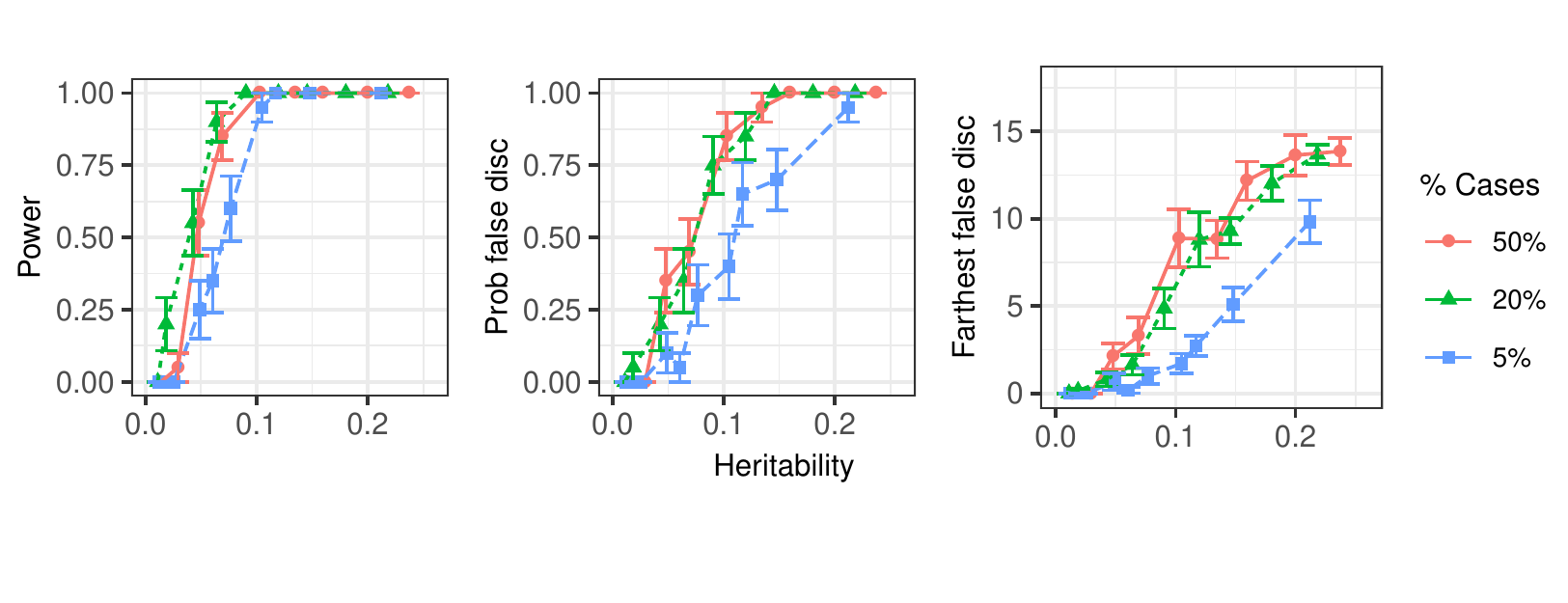}
\caption{Performance of the TDT in an admixed population. Left: power. Middle: probability of reporting at least one spurious discovery at least 1 Mb away from the causal SNP. Right: distance from the true causal SNP to the farthest spurious discovery.} 
\label{fig:admixed_power_fd}
\end{figure}

\begin{figure}[p]
\centering
\includegraphics[width = 6in]{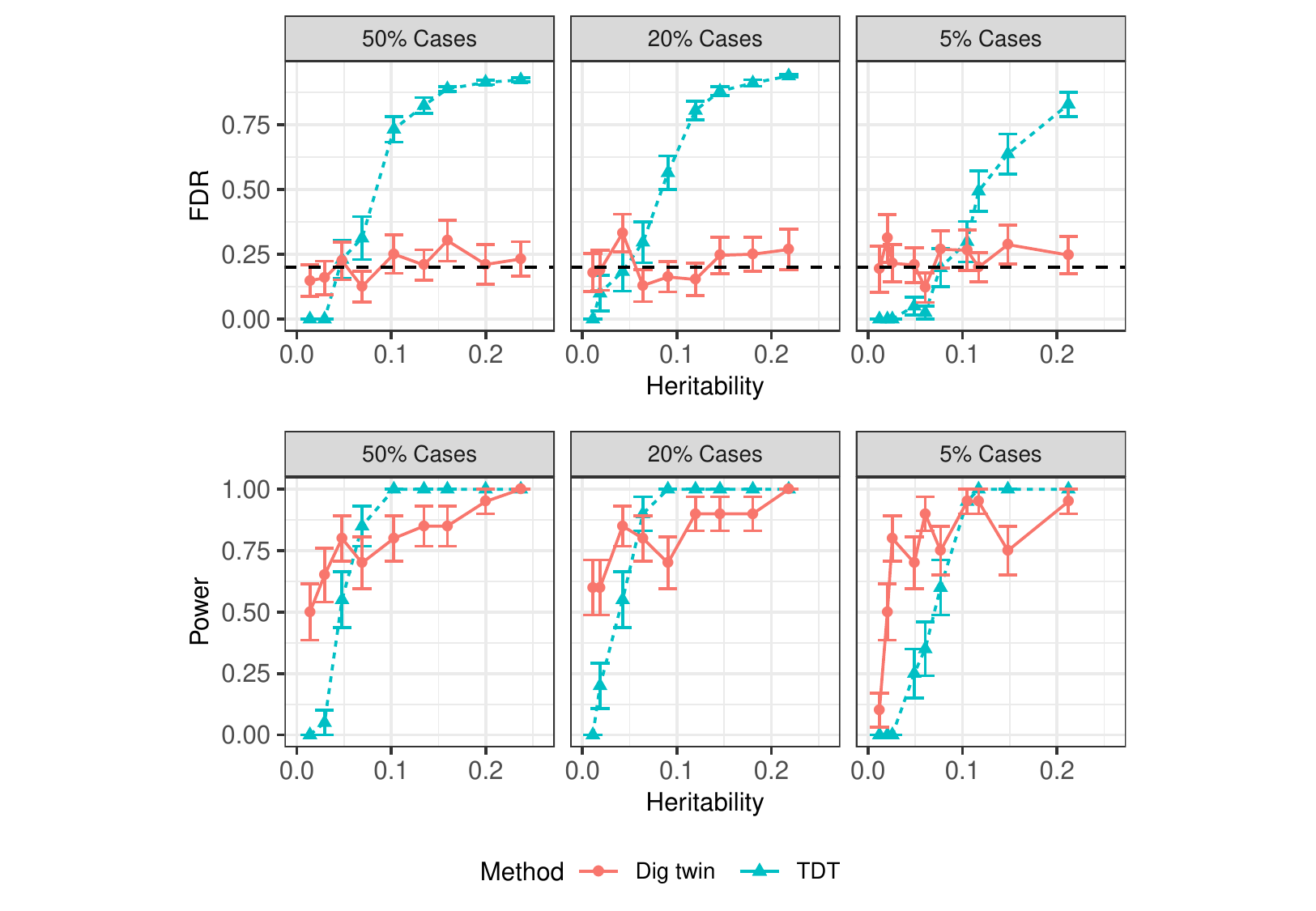}
\caption{Performance of the Local Digital Twin Test and TDT in an
  admixed population. In the top panel, the dashed horizontal line
  indicates the nominal FDR level for the Local Digital Twin
  Test. Because the TDT is using the genome-wide significance level, the nominal FDR level for the TDT is less than
  $0.05$.}
\label{fig:admixed_fdr}
\end{figure}

For simplicity, we perform the TDT at each SNP and report a SNP as significant if the p-value is below the genome-wide significance level, $5 \cdot 10^{-8}$. We use this demanding threshold to emphasize that even the most conservative existing methods violate type-I error control here. Indeed, recall that in the left-hand panel of Figure~\ref{fig:admixed_manhattan}, we find rejections all across the chromosome, even though there is only one causal SNP. If the heritability is large, Figure~\ref{fig:admixed_power_fd} shows that spurious associations at least 1 Mb away from the causal SNP occur with high probability; here, each experiment is repeated 20 times. Moreover, spurious associations occur as far away as 15 Mb from the causal SNP, reinforcing that one should not interpret the TDT p-values as providing evidence of distinct causal findings.

Proceeding as in Section~\ref{sec:sims_large_scale}, we split the chromosome into 25 groups of size 2 Mb and consider the procedure that reports a group as significant if any SNP therein has a p-value below the genome-wide significance threshold. Even though it may be tempting to expect this method to control the family-wise error rate (and by extension, the FDR), the results from Figure~\ref{fig:admixed_power_fd} suggest that this is not the case, and we confirm this directly in Figure~\ref{fig:admixed_fdr}. Similarly, we conclude that all of the other less conservative benchmarks from Section~\ref{sec:sims_large_scale} will fail to control the FDR in this setting. By contrast, the Local Digital Twin Test does not suffer from this problem because its p-values are valid for the conditional null hypotheses in \eqref{eq:cond_null}. Indeed, Figure~\ref{fig:admixed_fdr} shows that the Local Digital Twin Test controls the FDR at the nominal level.

\section{Analysis of autism spectrum disorders}

We apply the Digital Twin Test to study the genetic basis of autism
spectrum disorder (ASD) using a data set of 2,565 parent-child trios
from the Autism Sequencing Consortium (ASC) \citep{Buxbaum2012},
accessed through dbGaP \citep{Mailman2007, Tryka2014}; see Appendix~\ref{app:autism-analysis-details} for details about the sample and data processing. ASD has complex genetic roots, with variability arising from common SNPs, copy number variation, and de novo gene-disrupting mutations \citep{Iossifov2014,Gaugler2014}. It has been theoretically conjectured \citep{Devlin2011} and empirically observed \citep{Anney2012} that common variants only have small effects on ASD. As a result, only recently have individual SNPs been implicated by GWAS, and the first individual SNPs appear in \citep{Grove2019}, where 5 SNPs were reported at the genome-wide significance level in a sample of approximately 45,000 individuals. Since one of these five SNPs, {\em rs910805} (an intergenic variant on chromosome 20), was also genotyped in the ASC trio data, we will test its causal validity with the Local Digital Twin Test. Note, however, that the ASC trio data was used as part of \citep{Grove2019}, so our results here must be viewed as a practical demonstration and not as an independent replication of these findings. 

Since we do not have a corresponding cohort of non-trio observations, we use the Local Digital Twin Test with the test statistic in \eqref{eq:tdt_test_stat} and report the results in Table~\ref{tab:asd-conditional-tdt}. We test groups centered around SNP rs910805 of increasing size, ranging from 1 Mb to the whole chromosome (in which case this Local Digital Twin Test is equivalent to the TDT). For large enough groups, we find significance at the $\alpha = 0.05$ level---we are only testing one SNP so we do not need to achieve the genome-wide significance level. We are unable to reject at finer resolutions; smaller groups correspond to fundamentally more demanding statistical hypotheses, and the effect size here is small. Our analysis suggests that the observed association with rs910805 is not due to confounding and that there is instead a causal variant in its vicinity.

\begin{table}[t]
\begin{center}
 \begin{tabular}{|c| c c c c c c |} 
 \hline
 Resolution & 1 Mb & 2 Mb & 3 Mb & 4 Mb & 5 Mb & full-chromomsome \\ 
 \hline
 p-value & 0.237 & 0.146 & 0.100 & 0.0168 & 0.0244 & 0.011 \\ 
 \hline
\end{tabular}
\end{center}
\caption{Analysis of ASD with Local Digital Twin Test at different resolutions.}
\label{tab:asd-conditional-tdt}
\end{table}

\section{Discussion}

We have developed a flexible statistical test to rigorously establish
that a genomic region contains causal variants. It requires no
modeling assumptions because the inferences are based only on
variation arising in meiosis, a randomized experiment performed by
nature.  Nonetheless, we highlight here the two most important
limitations of this work. First, our method currently relies on
computationally phased haplotypes, as discussed in
Section~\ref{sec:setting}. Even though phasing is typically accurate,
especially when family data are available, it would be interesting to
thoroughly understand the robustness to phasing errors. Second, in
order to operate at practical resolutions, our method requires many
parent-child trios or duos, which are more challenging to collect than
unrelated individuals. Fortunately, larger data sets will become more
abundant as sequencing costs are rapidly decreasing
\citep{NIHSeequence}. Moreover, contemporary studies genotyping entire
populations (e.g., in Iceland \cite{deCODE} and Finland
\cite{finngen}) are particularly promising because they simultaneously
address the challenges of accurate phasing and obtaining data from
both parents and offspring. In such studies, drawing conclusions from
GWAS that are provably immune to confounding variables is appealing,
making the Digital Twin Test a tantalizing way to confirm candidate
GWAS discoveries.

Turning to the statistical aspect of this work, our finite-sample
causal conclusions come from the fact that (i) $A$ blocks the relevant confounders, and (ii) we exactly know the distribution of $X$ given
$A$. This statistical form may appear in other settings. Concretely,
for any random variables $X,Y$ and $Z$ such that the relationship in
Definition~\ref{def:ext_confounder} holds, Theorem~\ref{thm:causality}
holds as well. More generally, for a structural equation model on vectors $(A,X,Y,Z)$, if $A$ satisfies the back-door criterion \citep{pearl2009} with respect to $X$ and $Y$, then tests of the hypothesis \eqref{eq:cond_null} only detect causal effects. One can then leverage conditional testing techniques,
such as the conditional randomization test or knockoffs, to make
causal inferences with finite-sample guarantees.

\section*{Acknowledgements}
S.~B. is partly supported by a Ric Weiland Fellowship.
S.~B., M.~S., C.~S., and E.~C.~are supported by NSF grant DMS 1712800.
C.~S.~and E.~C.~are also supported by a Math+X grant (Simons Foundation) and by NSF grant 1934578.
We thank Trevor Hastie, Manuel Rivas, Robert Tibshirani, and Wing Wong for discussions of an early version of this manuscript and the Stanford Research Computing Center for computational resources and support.

\newpage
\bibliographystyle{ieeetr}
\bibliography{causal_snps.bib}

\newpage
\appendix

\renewcommand\thefigure{S.\arabic{figure}}
\setcounter{figure}{0}   

\section{Constructing independent p-values}
\label{app:independent_p}
We present a technical refinement of the Local Digital Twin Test (still testing the null hypothesis in \eqref{eq:group_null}) that produces independent p-values for disjoint null groups. The core idea is that for each individual, we determine which groups contain a recombination event; then, we generate the digital twins only in those groups, sampling from a modified version of the distribution \eqref{eq:ldtt_imputation_dist}. By construction, these digital twins will have recombination events in the same groups as the true observations, but the recombination events will be randomly perturbed. This separates the randomness used to test each region. We then show how to create feature importance statistics that ignore
the randomness in other regions. As a result, the tests in different regions are fully decoupled.

To begin, let $G$ denote the collection of disjoint groups of SNPs that we wish to test. We will assume that each group $g$ is a continuous block of the form $\{g_-, g_- + 1, \dots, g_+\}$ for endpoints $g_- \le g_+$ in $\{1,\dots,p\}$. Also, let $\mathcal{B} \subset \{1, \dots, p\}$ be the set of SNPs that form the boundary of some group: $\mathcal{B} = \{j : j = g_- \text{ or } j = g_+ \text{ for some } g\in G\}$.

For simplicity, we consider a single observation and define $U^m \in \{a, b\}^p$ and $U^f \in \{a, b\}^p$ to be the underlying ancestral states for the two haplotypes $X^f$ and $X^m$. We first sample from the posterior distribution conditional on the data based on the HMM described in Section \ref{sec:recomb_model} and then condition on $U^m_\B$ and $U^f_\B$.  This conditioning splits the distribution of $(X_g^f, U_g^f)$ and $(X_g^m, U_g^m)$ into {\em independent} HMMs across the groups $g \in G$, see \cite{Bates2019} for a general discussion of this conditioning idea. In the modified Local Digital Twin Test, each digital twin will be sampled from the distribution of
\begin{equation}
(X^f_g, X^m_g) \mid (U^m_\B, U^f_\B, A).
\label{eq:mldtt_imputation_dist}
\end{equation}
This is essentially the same as in the original Local Digital Twin Test, which instead samples from \eqref{eq:ldtt_imputation_dist}, because the values $U^m_\B, U^f_\B$ are almost perfectly determined by the data $X^f_{-g}$ and $X^m_{-g}$, since the measured SNPs are dense and recombination is rare. Sampling from the distribution \eqref{eq:mldtt_imputation_dist} is straightforward: we simply sample $U^m_g$ and $U^f_g$, an HMM sampling operation, and then from this we sample $X^m_g$ and $X^f_g$ from the emission distribution.

We now introduce additional notation that will be needed to describe
the digital twin sampling. Given $U^{f, \i}$ and $U^{m, \i}$, for each
observation $i=1,\dots,n$, let $R^{m, \i} = \{g \in G: U^{m,\i}_{g_-}
\ne U^{m,\i}_{g_+}\}$ be the set of groups where individual $i$ has an
odd number of recombinations (for small groups, there will typically
be exactly one recombination). Define $R^{f, \i}$ in the analogous
way. Now, let $\fixed^m = \{(i,j) :  j \notin g \text{ for all } g \in
R^{m, \i}\}$ be the set of observations and sites in groups with an even number of recombinations. Define $\fixed^f$ in the analogous way. The upcoming algorithm will hold $X^{m,\i}_j$ for $(i, j) \in \fixed^m$ fixed when generating the digital twins.

Next, we define the masked version of $X$ which is safe to use in the feature importance statistics; feature importance statistics using the masked version of $X$ will operate independently in different groups. The masked version is defined as follows:
\begin{align}
\label{eq:masked_x}
\begin{split}
X^\maskmi_j = 
\begin{cases}
X^\mii_g & \text{ if } (i,j) \in \fixed^m, \\
E[X^\mii_g \mid M_a, M_b, U^m_\B]  & \text{ otherwise},
\end{cases} \\
X^\maskfi_j = 
\begin{cases}
X^\fii_g & \text{ if }  (i,j) \in \fixed^f, \\
E[X^\fii_g \mid M_a, M_b, U^f_\B]  & \text{ otherwise},
\end{cases} \\
\end{split}
\end{align}
for each observation $i = 1,\dots,n$. The entries of $X^m$ in groups with no recombination events will remain unchanged in both the digital twins and in $X^\maskmi$.
On the other hand, for groups with observed recombinations, the digital twins will vary, so, for these entries, we set $X^\maskm$ to be a constant: the average imputation based on the parental haplotypes and $(U^m_\B, U^f_\B)$. These conditional expectations can be computed easily, thanks to the Markov property. Notice that, since there are very few groups with recombination events, most entries of $X^\maskm$ will be equal to those of $X^m$, and most entries of $X^\maskf$ will be equal to those of $X^f$. For instance, with groups of size 1 Mb we will have that over $99\%$ of the entries match. 

With the notation in place, we can finally present the procedure in Algorithm~\ref{alg:ldtt-independent}. \\
\begin{algorithm}[H]
\SetAlgoLined
\DontPrintSemicolon
--- Sample $U^m$ and $U^f$ given $(X^m, X^f, A)$ according the HMM in Section~\ref{sec:recomb_model} (see Appendix~\ref{sec:sampling_digital_twins} for an explicit sampler). \;
--- Define $X^\maskm$ and $X^\maskf$ according to \eqref{eq:masked_x}.\;

\For{$g \in G$}{
--- Compute the test statistic on the true data: $$t^* = T(X^\maskm_{-g} + X^\maskf_{-g}, X^m_{g} + X^f_{g}, Y).\;$$
\For{$k = 1, \dots, K$}{
	--- For observations $i$ such that $g \in R^{\mii}$, sample $\wtilde{X}^\mii_g$ from the distribution \eqref{eq:mldtt_imputation_dist}, independent of $X^f, X^m, U^f, U^m$ and $Y$ (see Appendix~\ref{sec:sampling_digital_twins} for an explicit sampler). Otherwise, set $\wtilde{X}^\mii_g = X^\mii_g$. Sample $\Xk^f_g$ analogously. \vspace{0.05in} \;
	--- Compute the test statistic on the digital twins: 
	$$t_k = T(X^\maskm_{-g} + X^\maskf_{-g}, \Xk^m_{g} + \Xk^f_{g}, Y).\;$$
}
--- Compute the quantile of the true statistic $t^*$ among the digital twin statistics $t_1, \dots, t_K$:
$$v_g = \frac{1 + \#\{k: t^* \le t_k\}}{ K+ 1}.\;$$
}

\KwResult{A family of independent p-values $\{v_g\}_{g \in G}$ for
  testing the conditional hypotheses  \eqref{eq:group_null}.}
\caption{The Local Digital Twin Test with independent p-values}
\label{alg:ldtt-independent}
\end{algorithm}

The null p-values produced by this algorithm are jointly independent, which we record in Theorem~\ref{thm:independent_pvals} in Section~\ref{sec:cond_dtt}. The proof is given in Appendix~\ref{app:proofs}.

\paragraph{Remark} This procedure can be viewed as an analogue to the model-X
knockoffs procedure \citep{candes2016}, except with $K$ knockoffs instead of the usual two. Conditional on $U_\B$, the Markov structure ensures that the imputations (i.e. the multiple knockoffs) are jointly exchangeable and hence are analogous to multiple knockoffs. Masking the values of $X$ where the imputations take place is analogous to making the feature importance statistic for group $g$ invariant to swapping $X_{g'}$ and $\Xk_{g'}$ for all groups $g'$ in the original knockoff setting.

\section{Proofs}
\label{app:proofs}

We first give an explicit proof for Proposition~\ref{prop:chromosome_dtt_causality}. The argument is essentially the same as that used in Section~\ref{sec:causality} to prove Theorem~\ref{thm:causality}.
\begin{proof}[Proof of Proposition~\ref{prop:chromosome_dtt_causality}]
Suppose the null hypothesis $H_0$ is true, so that 
\begin{equation*}
Y \indep X_\C \mid (A, Z).
\end{equation*}
Then, since $Z$ is an external confounder, 
\begin{equation*}
Z \indep X_\C \mid A,
\end{equation*}
and so
\begin{equation*}
(Y, Z) \indep X_\C \mid A.
\end{equation*}
Thus, the following null hypothesis is also true:
\begin{equation*}
H_0': Y \indep X_\C \mid A.
\end{equation*}
The Digital Twin Test is a valid test of $H_0'$, since it is an instance of the conditional randomization test \citep{candes2016}, and so the result follows.
\end{proof}

Next, we turn to Proposition~\ref{prop:ldtt_causal_validity}. The proof is essentially the same as the above.
\begin{proof}[Proof of Proposition~\ref{prop:ldtt_causal_validity}]
Suppose the null hypothesis $H_0^g$ is true so that 
\begin{equation*}
Y \indep X_g \mid (X^m_{-g}, X^f_{-g}, A, Z).
\end{equation*}
Then, since $Z$ is an external confounder, we have that
\begin{equation*}
Z \indep X_g \mid (X^m_{-g}, X^f_{-g}, A),
\end{equation*}
and so
\begin{equation*}
(Y, Z) \indep X_g \mid (X^m_{-g}, X^f_{-g}, A).
\end{equation*}
Thus, the following null hypothesis is also true:
\begin{equation*}
H_0^{\prime g}: Y \indep X_g \mid (X^m_{-g}, X^f_{-g}, A).
\end{equation*}
The Local Digital Twin Test is a valid test of $H_0^{\prime g}$, since it is an instance of the conditional randomization test \citep{candes2016}, and so the result follows.
\end{proof}

Lastly, we give a proof that the p-values from the modified Local Digital Twin Test in Appendix~\ref{app:independent_p} are super-uniform and independent. 

\begin{proof}[Proof of Theorem~\ref{thm:independent_pvals}]
First, note that any group $g$ satisfying the Local Digital Twin Test null hypothesis in \eqref{eq:group_conf_null} satisfies the weaker null hypothesis in \eqref{eq:group_null} by the argument in the proof of Proposition~\ref{prop:ldtt_causal_validity}. As a result, it suffices to show the result for the case without any external confounder $Z$.
 
Let $n_G = |G|$ be the number of groups and recall that $K$ is the number of iterations in Algorithm~\ref{alg:ldtt-independent}. Conditional on $(U^m_\B, U^f_\B, A)$, we have the following simultaneous exchangeability property:
\begin{multline}
\left((X^m_{g_1}, \Xk^{m, 1}_{g_1}, \dots, \Xk^{m, K}_{g_1}), \dots, 
 (X^m_{g_{n_G}}, \Xk^{m, 1}_{g_{n_G}}, \dots, \Xk^{m, K}_{g_{n_G}})\right) \eqd \\ 
 \left( (X^m_{g_1}, \Xk^{m, 1}_{g_1}, \dots, \Xk^{m, K}_{n_G})_{\sigma_1}, \dots,
 (X^m_{g_{n_G}}, \Xk^{m, 1}_{g_{n_G}}, \dots, \Xk^{m, K}_{g_{n_G}})_{\sigma_{n_G}} \right),
 \label{eq:mldtt_xm_exch}
\end{multline}
for any permutations $\sigma_1, \dots, \sigma_{n_G}$ of $(1,\dots,K+1)$. This is true because conditional on $(U^m_\B, U^f_\B, A)$, the $X^m_g$ for different groups $g$ are independent HMMs. Moreover, \eqref{eq:mldtt_xm_exch} continues to hold if we additionally condition on $(X^\maskm, X^\maskf)$, since this is equivalent to conditioning on more entries of $U^f$ and $U^m$, after which the $X^m_g$ remain independent HMMs. The analogous property holds for $X^f$.

Next, we show that \eqref{eq:mldtt_xm_exch} continues to hold for the null groups even after additionally conditioning on $Y.$ Without loss of generality, suppose $g_1$ is a null group. First, by construction of $(U^m, U^f)$ we have that
\begin{equation*}
Y \indep (U^m, U^f) \mid X^f, X^m, A.
\end{equation*} 
As a result,
\begin{equation*}
Y \indep X_{g_1} \mid (X^m_{-g_1}, X^f_{-g_1}, A) \implies Y \indep (X_{g_1}, U^m, U^f) \mid (X^m_{-g_1}, X^f_{-g_1}, A).
\end{equation*} 
The right-hand side of the above further implies that 
\begin{equation*}
Y \indep X_{g_1} \mid (X^m_{-g_1}, X^f_{-g_1}, X^\maskm, X^\maskf, U^f_\B, U^m_\B, A),
\end{equation*} 
since $(X^\maskm, X^\maskf, U^f_\B, U^m_\B)$ is a function of $(X^m_{g_1}, X^f_{g_1}, X^m_{-g_1}, X^f_{-g_1}, U^f, U^m, A)$. 
By the independence in the construction of the digital twins, we then have
\begin{multline*}
Y \indep X_{g_1} \mid \bigr((X^m_{g_2}, \Xk^{m, 1}_{g_2}, \dots, \Xk^{m, K}_{g_2}), \dots, 
 (X^m_{g_{n_G}}, \Xk^{m, 1}_{g_{n_G}}, \dots, \Xk^{m, K}_{g_{n_G}}), \\
 X^\maskm, X^\maskf, U^f_\B, U^m_\B, A \bigl),
\end{multline*}
and further 
\begin{multline*}
Y \indep (X^m_{g_1}, \Xk^{m, 1}_{g_1}, \dots, \Xk^{m, K}_{g_1}) \mid \biggl( (X^m_{g_2}, \Xk^{m, 1}_{g_2}, \dots, \Xk^{m, K}_{g_2}), \dots, 
 (X^m_{g_{n_G}}, \Xk^{m, 1}_{g_{n_G}}, \dots, \Xk^{m, K}_{g_{n_G}}), \\
 X^\maskm, X^\maskf, U^f_\B, U^m_\B, A\biggr).
\end{multline*}
We conclude that \eqref{eq:mldtt_xm_exch} holds conditional on $(X^\maskm, X^\maskf, U^f_\B, U^m_\B, A, Y)$ for any permutations $\sigma_1, \dots, \sigma_{n_G}$ of $(1,\dots,K+1)$ such that the non-null groups have the identity permutation. The analogous statement holds for $X^f$. Combining this with the analogous result for $X^f$, we then have:
\begin{multline}
\left((X_{g_1}, \Xk^{1}_{g_1}, \dots, \Xk^{K}_{g_1}), \dots, 
 (X_{g_{n_G}}, \Xk^{1}_{g_{n_G}}, \dots, \Xk^{K}_{g_{n_G}})\right) \eqd \\ 
 \left( (X_{g_1}, \Xk^{1}_{g_1}, \dots, \Xk^{K}_{n_G})_{\sigma_1}, \dots,
 (X_{g_{n_G}}, \Xk^{1}_{g_{n_G}}, \dots, \Xk^{K}_{g_{n_G}})_{\sigma_{n_G}} \right),
 \label{eq:mldtt_x_exch}
\end{multline}
conditional on $(X^\maskm, X^\maskf, U^f_\B, U^m_\B, A, Y)$ for any permutations $\sigma_1, \dots, \sigma_{n_G}$ of $(1,\dots,K+1)$ such that the non-null groups have the identity permutation.

We now turn to the feature importance statistics $t^*, t_1, \dots, t_K$ for group $g$. These are a function of $(X^m_{g}, \Xk^{m, 1}_{g}, \dots, \Xk^{m, K}_{g})$ and $(X^f_{g}, \Xk^{f, 1}_{g}, \dots, \Xk^{f, K}_{g})$ together with $(X^\maskm, X^\maskf, U^f_\B, U^m_\B,$ $A, Y)$, and we are conditioning on the latter. In particular, the feature importance statistics for group $g$ do not depend on $(X^m_{g'}, \Xk^{m, 1}_{g'}, \dots, \Xk^{m, k}_{g'})$ for any other group $g'$, except through the entries already fixed in $X^\maskm$. Combining this fact with the result \eqref{eq:mldtt_x_exch}, we have that
\begin{equation*}
\left((t^*_{g_1}, t^{1}_{g_1}, \dots, t^{K}_{g_1}), \dots, 
 (t^*_{g_{n_G}}, t^{1}_{g_{n_G}}, \dots, t^{K}_{g_{n_G}}) \right) \eqd 
 \left( (t^*_{g_1}, t^{1}_{g_1}, \dots, t^{K}_{g_1})_{\sigma_1}, \dots,
 (t^*_{g_{n_G}}, t^{1}_{g_{n_G}}, \dots, t^{K}_{g_{n_G}})_{\sigma_{n_G}} \right),
\end{equation*}
for any permutations $\sigma_1, \dots, \sigma_{n_G}$ of $(1,\dots,K+1)$, provided the non-null groups have the identity permutation, conditional on $(X^\maskm, X^\maskf, U^f_\B, U^m_\B, A, Y)$. This implies that conditionally the p-values $v_g$ for all null groups $g$ are jointly independent and have the discrete uniform distribution (if we randomly break ties), and so the same holds after marginalizing out $(X^\maskm, X^\maskf, U^f_\B, U^m_\B, Y)$.
\end{proof}

\section{Sampling digital twins}
\label{sec:sampling_digital_twins}
In this section, we explicitly describe how to sample the digital twins in Algorithms~\ref{alg:global_dtt}-\ref{alg:ldtt-independent}. In each case, the sampling is a standard conditional sampling operation on an HMM \citepeg{Jurafsky2000}, but we give the details for our setting for completeness.

For convenience, we consider a single observation so that $X^m, X^f, M^a, M^b,$ $F^a, F^b \in \{0,1\}^p$ are the haplotypes, and whose hidden ancestral states are $U^m, U^f \in \{a,b\}^p$. Without loss of generality, we consider one chromosome; different chromosomes are independent, so this is sufficient. In addition, for $u \in \{a, b\}$ we define:
\begin{equation*}
\overline{u} = \begin{cases}
a & \text{ when } u = b, \\
b & \text{ when } u = a,
\end{cases}
\end{equation*}
i.e., the complementary ancestral strand.

\subsection{Global digital twins}
This subsection describes how to sample the digital twins in Algorithm~\ref{alg:global_dtt}. Without loss of generality, we consider $X^m$, sampled as follows:
\begin{enumerate}
\item (Hidden ancestral states) Sample $u^m_1$ uniformly from $\{a,b\}$.
\item For $j=1,\dots,p$, independently sample $u^m_j$ as
\begin{equation*}
u_{j}^m = 
\begin{cases}
u_{j-1}^m & \text{with probability $1/2 (1 + e^{-2 d_j})$}, \vspace{.1in}\\
\overline{u_{j-1}^{m}} & \text{ otherwise}.
\end{cases}
\end{equation*}

\item (De novo mutations) For each $j=1,\dots,p$, independently sample $X_j^m$ as
\begin{equation*}
X_{j}^m = \begin{cases}
M_{j}^{u_j^{m}} & \text{with probability $1 - \epsilon$}, \\
1 - M_{j}^{u_j^{m}} & \text{otherwise}.
\end{cases}
\end{equation*}
\end{enumerate}

\subsection{Local digital twins}
\label{app:ldtt_sampler}

This subsection shows how to sample the digital twins in Algorithm~\ref{alg:ldtt}. Again, we consider only $X^m$ explicitly.

\subsubsection*{Computing forward-backward weights} 
We start by computing the forward-backward weights for the HMM from Section~\ref{sec:recomb_model}. Recall that our model gives the following quantities:
\begin{equation*}
P(U^m_j = u^m_j \mid U^m_{j-1} = u^m_{j-1}) = 
\begin{cases}
\frac{1}{2} (1 + e^{-2 d_j}) & \text{ if $u^m_j = u^m_{j-1}$}, \\
1 - \frac{1}{2} (1 + e^{-2 d_j}) & \text{otherwise},
\end{cases}
\end{equation*}
\begin{equation*}
P(X^m_j = x^m_j \mid U^m_j = u^m_j) = 
\begin{cases}
1 - \epsilon & \text{ if $M^{u^m_j}_1 = X^m_j$}, \\
\epsilon & \text{otherwise},
\end{cases}
\end{equation*}
and
\begin{equation*}
P(U^m_j = U^m_{j'}) = \frac{1}{2} (1 + e^{-2(d_j+\dots+d_{j'-1})}), \text{ \quad for $j' > j$}.
\end{equation*}
We define the {\em forward weights} for $j=1,\dots,p$ as
\begin{equation*}
\alpha^m_j(u^m_j) = P(X^m_{1:j} = x^m_{1:j}, U^m_j = u^m_j),
\end{equation*}
where the notation $x_{j:j'}$ indicates the subvector of $x$ of entries $j,j+1,\dots, j'$.
The $\alpha^m_j(\cdot)$ can be computed recursively as follows:
\begin{align*}
\alpha^m_1(u^m_1) &= \begin{cases}
\frac{1}{2} (1 - \epsilon)  & \text{ if $M^{u^m_1}_1 = X^m_1$}, \\
\frac{1}{2} \epsilon & \text{otherwise},
\end{cases} \\
\alpha_j^m(u_j^m) &= \sum_{u \in \{a,b\}} \alpha^m_{j-1}(u) P(U^m_j = u^m_j \mid U^m_{j-1} = u) P(X^m_j = x^m_j \mid U^m_j = u^m_j),
\end{align*}
for $j=2,\dots,p$ and $u_j^m \in \{a,b\}$. For convenience, we also define $\alpha_0(u) = 1/2$ for $u \in \{a,b\}$.

Next, we define the {\em backward weights} as
\begin{equation*}
\beta_j^m(u_j^m) = P(X_{(j+1):p} = x_{(j+1):p}\mid U_j^m = u^m_j),
\end{equation*}
with the convention that $\beta_p^m(u) = 1$ for $u \in \{a, b\}$.
These can be computed recursively as follows:
\begin{equation*}
\beta_{j}(u_j^m) = \sum_{u \in \{a, b\}} \beta_{j+1}^m(u) P(U_{j+1}^m = u \mid U_{j}^m = u_j^m) P(X_{j+1}^m = x_{j+1}^m \mid U_{j+1}^m = u),
\end{equation*}
for $j = p-1,\dots,1$ and $u^j_m \in \{a, b\}$.

\subsubsection*{Sampling the local digital twins}
Consider a group of SNPs $g = \{g_-, g_- + 1, \dots, g_+\} \subset \{1,\dots,p\}$. We sample a local digital twin as follows:

\begin{enumerate}
\item (Hidden ancestral states) Sample $U^m_{g_-} \in \{a, b\}$ from the distribution
\begin{multline*}
P(U^m_{g_-} = u) \propto 
P(X_{g_-}^m = x_{g_-}^m \mid U_{g_-}^m = u)
\left[P(U^m_{g_+} = U^m_{g_-}) \beta_{g_+}(u) + P(U^m_{g_+} \ne U^m_{g_-}) \beta_{g_+}(\overline{u})\right] \\
\left[\alpha_{g_- - 1}(u)P(U^m_{g_-} = U^m_{g_- - 1}) +  \alpha_{g_- - 1}(\overline{u})P(U^m_{g_-} \ne U^m_{g_- - 1}) \right]
.
\end{multline*}
\item For $j=g_- + 1,\dots,g_+$, sample $U^m_{j} \in \{a, b\}$ from the distribution
\begin{multline*}
P(U_j^m = u) \propto P(X_j^m = x_j^m \mid U_j^m = u) P(U_j^m = u \mid U_{j-1}^m = u_{j-1}^m) \\  \left[P(U_{g_+} = U_{j}) \beta_{g_+}(u) + P(U_{g_+} \ne U_{j}) \beta_{g_+}(\overline{u})\right].
\end{multline*}
\item (De novo mutations) For each $j=g_-,\dots,g_+$, independently sample $\Xk_j^m$ as
\begin{equation*}
\Xk_{j}^m = \begin{cases}
M_{j}^{u_j^{m}} & \text{with probability $1 - \epsilon$}, \\
1 - M_{j}^{u_j^{m}} & \text{otherwise}.
\end{cases}
\end{equation*}
\end{enumerate}

\subsection{Modified local digital twins}
This subsection describes how to sample the digital twins in Algorithm~\ref{alg:ldtt-independent}. As usual, we consider only $X^m$ explicitly. First, we sample the hidden ancestral states $U^{m} \in \{a, b\}^p$ using the sampler from Appendix~\ref{app:ldtt_sampler}, except with group $g = \{1,2,\dots,p\}$.

Next, we focus on a group $g = \{g_-, g_- + 1, \dots, g_+\} \subset \{1,\dots,p\}$. If $U^m_{g_-} = U^m_{g_+}$, then we set $\Xk_g^m = X_g^m$; otherwise, we sample the digital twin as follows: 
\begin{enumerate}
\item Set $\widetilde{U}^m_{g_-} = U^m_{g_-}$ and $\widetilde{U}^m_{g_+} = U^m_{g_+}$.
\item Repeat until $m$ is odd: sample $m \sim \text{Pois}(d_{g_- + 1} + \dots + d_{g_+})$ (i.e. sample from a Poisson conditional on the outcome being odd).
\item Sample the recombination points $r_1,\dots,r_m$ as independently from the discrete distribution supported on $\{g_- + 1,\dots,g_+\}$ with probability vector proportional to 
\begin{equation*}
\left( \frac{1}{2}(1 + e^{-2d_{g_- + 1}}), \dots, \frac{1}{2}(1 + e^{-2d_{g_+}})\right).
\end{equation*}
\item For $j=g_- + 1,\dots, g_+$, define $u^m_j$ as
\begin{equation*}
u_{j}^m = \begin{cases}
u_{j-1}^m & \text{if \#\{$k$: $r_k = j$\} is even}, \vspace{.1in}\\
\overline{u_{j-1}^{m}} & \text{otherwise}.
\end{cases}
\end{equation*}
\item (De novo mutations) For each $j=g_-,\dots,g_+ -1$, independently sample $\Xk_j^m$ as
\begin{equation*}
\Xk_{j}^m = \begin{cases}
M_{j}^{\tilde{u}_j^{m}} & \text{with probability $1 - \epsilon$}, \\
1 - M_{j}^{\tilde{u}_j^{m}} & \text{otherwise.}
\end{cases}
\end{equation*}
\end{enumerate}

\section{Potential outcomes formulation}
\label{app:potential_outcomes}

We can also formulate the claims of Section~\ref{sec:causality} in
terms of the potential outcomes framework \citep{rubin1974,
  rubin2005}. Suppose that all SNPs with a causal effect are
measured. This is not realistic in practical examples, but it is challenging to formulate the localization of unmeasured SNPs in the language of potential outcomes. In any case, this will be enough to demonstrate that our inferences are consistent with this point of view. Next, we assume that the potential outcomes for unit $i$ are only affected by the subject's genotypes $x^\i$ (the SUTVA assumption \citep{Rubin1980}) and let $Y_i(x^\i)$ denote the potential outcomes for unit $i$ under the ``treatment assignment'' $x^\i \in \{0,1,2\}^p$, and let $Y_i = Y_i(X^\i)$ denote the observed value of $Y_i$. Within this framework, the potential outcomes $Y_i(x^\i)$ are considered fixed (nonrandom) but unknown, whereas $X^\i$ and $Y_i$ are considered random. We wish to understand whether certain coordinates of our vector-valued treatment assignment have a causal effect. To formalize this, consider a group of SNPs $g$. We will test the sharp null hypothesis 
\begin{equation*} 
H_0^{g, \text{po}} : Y_i(x^\i) = Y_i(\tilde{x}^{\i}) \text{ for all $x^\i$ and $\tilde{x}^{\i}$ such that $x_{-g}^\i = \tilde{x}^{\i}_{-g}$, for every unit $i$}. 
\end{equation*} That is, the potential outcomes are the same whenever the SNPs outside the group $g$ are the same. Lastly, we assume that the treatment assignment is random, that is, that the $X_i$ are independent samples from the recombination model in Section~$\ref{sec:recomb_model}$. Under these assumptions
\begin{equation*}
H_0^{g, \text{po}} \implies Y \indep X_g \mid (X^m_{-g}, X^f_{-g}, A).
\end{equation*}
Notice that the right hand side is exactly the same as in $\eqref{eq:group_null}$, the null hypothesis tested by the Local Digital Twin Test. Thus, the Local Digital Twin Test is a valid test of $H_0^{g, \text{po}}$, and the resulting inferences are causal inferences in the potential outcomes sense.

\section{ASD data analysis details}
\label{app:autism-analysis-details}

The ASC data was accessed through the dbGaP platform through accession \href{https://www.ncbi.nlm.nih.gov/projects/gap/cgi-bin/study.cgi?study_id=phs000298.v4.p3}{phs000298.v4.p3.} It contains data for $7880$ individuals from $2611$ pedigrees regarding $947,233$ autosomal markers. After removing $11$ individuals that have high missingness and keeping only complete trios, we are left with $2,565$ trios. These genotypes are computationally phased using the \texttt{SHAPEIT} software \citep{Delaneau2012} incorporating the known pedigree information \citep{OConnell2014} and using the genetic distances from the 1000 genomes project \citep{Auton2015,Sudmant2015}.

\section{Permutation tests are not robust to confounders}
\label{sec:perm_failure}

Permutation tests are not a valid alternative to the Digital Twin Test for testing causal null hypotheses, as demonstrated by the following simple example.
Consider a population consisting of two distinct subpopulations such that, for some given SNP $X_j$, $X_j = 2$ in the first subpopulation and $X_j=0$ in the second subpopulation.
Suppose we are interested in a response $Y$ on which $X_j$ has no causal effect whatsoever, but that happens to satisfy $Y=1$ in the first subpopulation and $Y=0$ in the second subpopulation due to other genetic components or external confounders.
A permutation test would have no choice but to discover a significant association between $Y$ and $X_j$, even though there is no causal relation. By contrast, the Digital Twin Test would never reject the causal null hypothesis because the digital twins preserve the population structure and hence satisfy $\Xk_j=X_j$, as shown in Table~\ref{tab:permutation_illustration}. 

Indeed, it is widely known that permutation tests need to be corrected for population structure---a common approach would involve regressing the top principal components of the genetic matrix out of the response vector \citep{Price2006}. Although this correction may significantly mitigate the problem, it is approximate and it may not account for all other unknown confounders. By contrast, Theorem~\ref{thm:causality} tells us that the Digital Twin Test implicitly accounts for a very broad family of unknown confounders.

\begin{table}[H]
\begin{center}
\begin{tabular}{c c c c | c c}
 Subject & Population & Y & $X_j$ & $X_j^*$ & $\Xk_j$ \\ [0.5ex]
 \hline
 1 & 1 & 1 & 1 & 0 & 1 \\
 2 & 1 & 1 & 1 & 1 & 1 \\
 3 & 1 & 1 & 1 & 1 & 1 \\
 4 & 1 & 1 & 1 & 0 & 1 \\
 5 & 2 & 0 & 0 & 1 & 0 \\
 6 & 2 & 0 & 0 & 0 & 0 \\
 7 & 2 & 0 & 0 & 1 & 0 \\
 8 & 2 & 0 & 0 & 0 & 0 \\
 \hline
\end{tabular}
\end{center}
\caption{A hypothetical sample in which strong population structure creates a spurious association between the trait, $Y$, and some SNP, $X_j$. The population structure is not preserved by permutations ($X_j^*$), which  report non-causal associations, but is preserved by the Digital Twin Test ($\Xk_j$), which can only discover causal effects.}
\label{tab:permutation_illustration}
\end{table}

\

\section{Auxiliary simulations and experiment details}

\subsection{Heritability}
\label{app:heritability}
Throughout this paper, the results of the simulations with binary traits are indexed by the heritability, which is measured on the liability scale:
\begin{equation*}
  \frac{\text{var}(\hat{Y})}{\text{var}(Y)} \cdot \frac{\bar{Y} (1-\bar{Y})}{\phi(\Phi^{-1}(\bar{Y}))}.
\end{equation*}
Above, $\hat{Y}$ is the mean of $Y$ conditional on the genotypes $X$, according to the true model, $\bar{Y}$ is the true expectation of $Y$, $\phi$ is the density of the standard normal distribution and $\Phi$ is its cumulative distribution function. This scale is widely used because it makes the range of heritability values independent of the diseases prevalence, $\bar{Y}$ \citep{falconer1965}.

\subsection{Main text experiment details}
\label{sec:app_sims_large_scale}

For the experiment in Section~\ref{sec:global_null_dtt_sim}, we report metrics about the lasso fit $\hat{\beta}$ on the external GWAS in Figure~\ref{fig:dtt_global_null_betas}. 
%Here, each experiment is repeated 20 times.

\begin{figure}
\includegraphics[width=\textwidth]{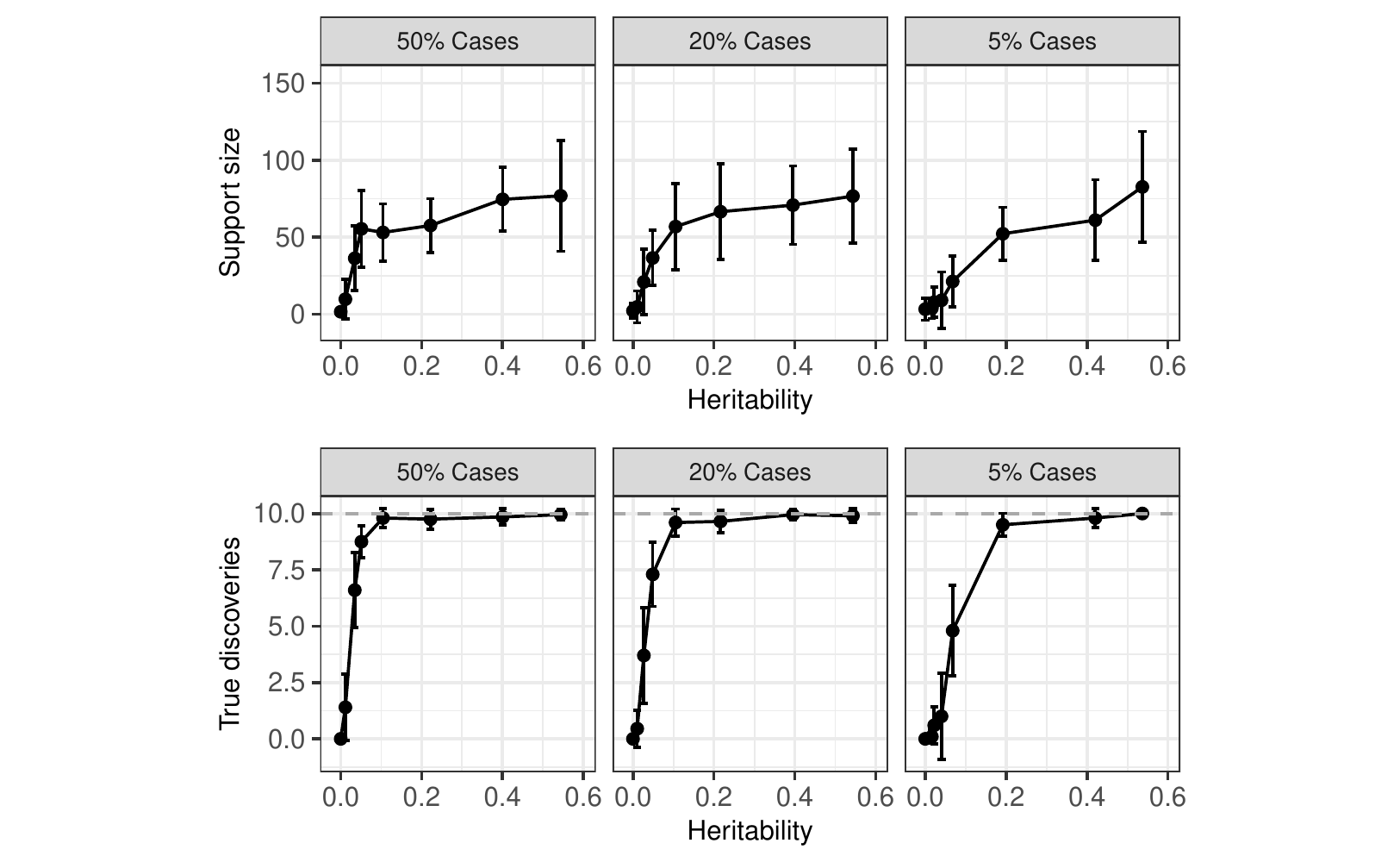}
\caption{Properties of the estimated regression coefficient $\hat{\beta}$ in the simulations in Section~\ref{sec:global_null_dtt_sim}. Error bars show one standard deviation.}
\label{fig:dtt_global_null_betas}
\end{figure}

Turning the experiment in Section~\ref{sec:sims_large_scale}, the hypotheses for the accumulation test are sorted in decreasing order with respect to $\sum_{j\in G_k} |\hat{\beta}_j^{\text{GWAS}}(\lambda_{\text{cv}})|$. We use an accumulation test with the hinge exponential function with $c = 2$ \citep{li2017}. 

\begin{figure}
  \centering
  \includegraphics[width=0.75\linewidth]{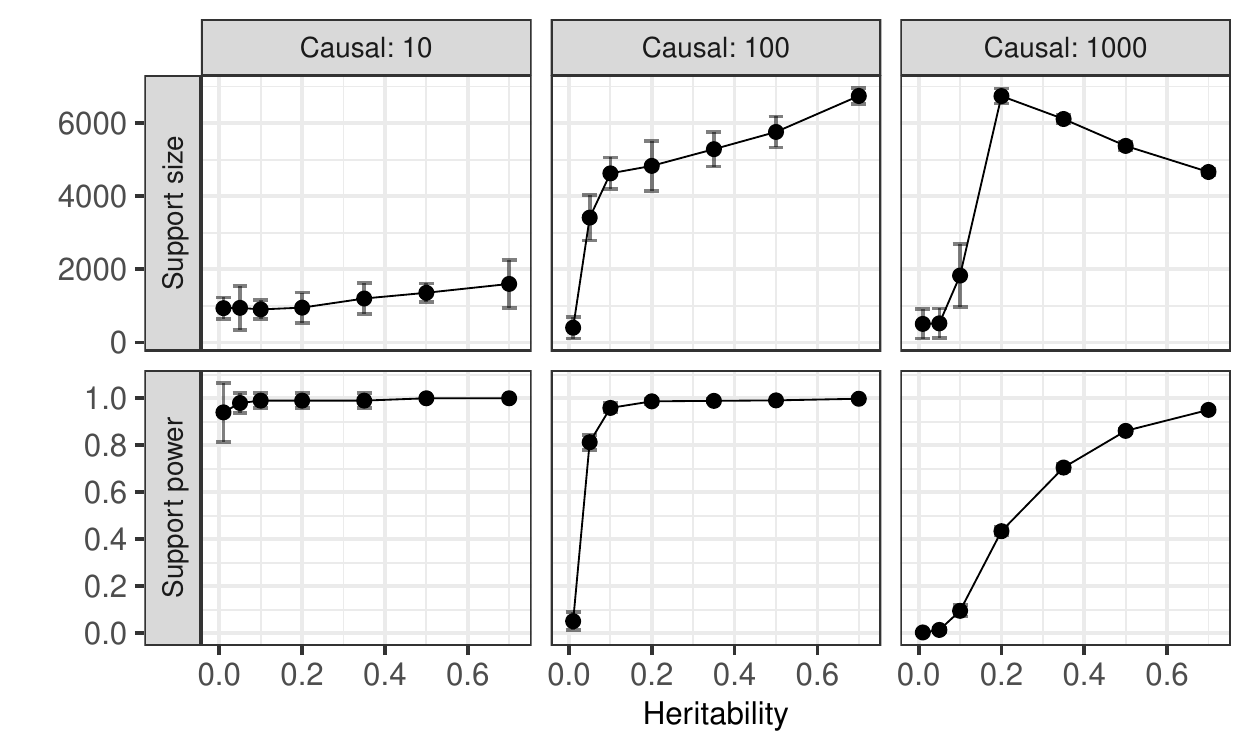}
  \caption{Properties of the regression coefficients $\hat{\beta}$ in the full-genome simulations with binary trait. Error bars show one standard deviation.}
  \label{fig:exp_large_binary_beta}
\end{figure}

As an added diagnostic, we plot the p-values from the Local Digital Twin Test and confirm that they are uniformly distributed in Figure~\ref{fig:exp_large_continuous_1Mb_histogram}. We confirm that the null p-values are super-uniform in distribution, and further that this holds even for the nulls that are identified as most promising by the external GWAS.\footnote{
	The parameter $\epsilon$ should be set to the value $10^{-8}$, the mutation rate for humans \citep{Acuna-Hidalgo2016}. We observed that 
	when we simulated offspring with $\epsilon = 0$ and then used a larger value (such as $\epsilon = 10^{-6}$) in the digital twin sampling, the null p-values had a non-uniform distribution.
}
\begin{figure}
  \centering
  \includegraphics[]{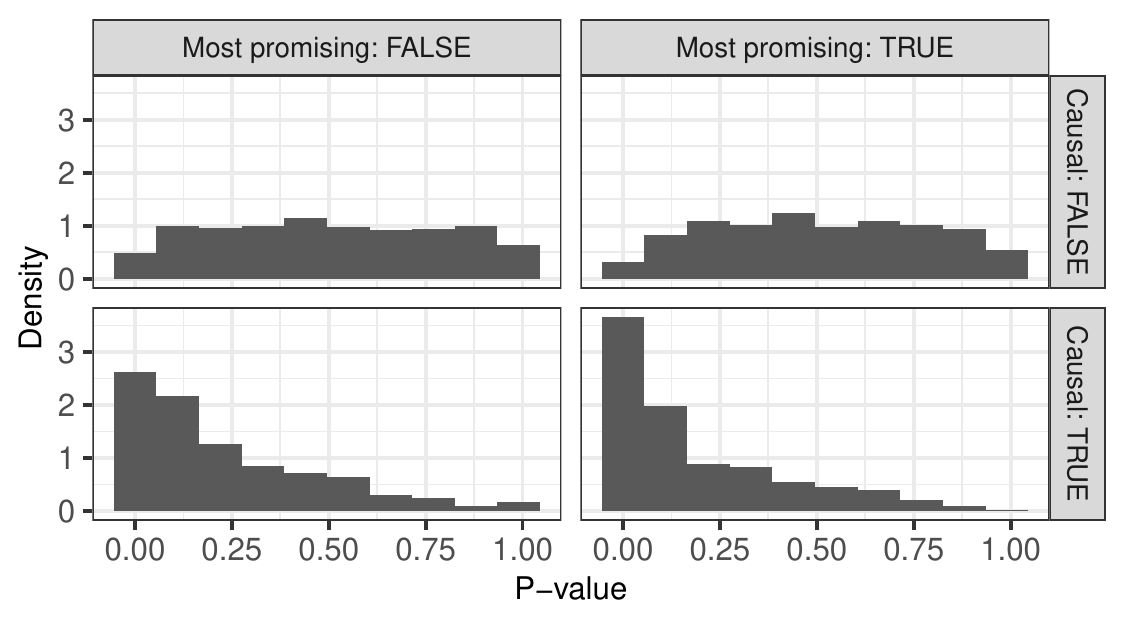}
  \caption{Histogram of null and non-null p-values computed by the
    local Digital Twin Test for a simulated binary trait with 100
    causal variables. The null p-values are stratified based on
    whether they correspond to genetic groups among the 200 most
    promising ones as inferred from the analysis of independent GWAS
    data. The total heritability of the simulated trait is 70\%.}
  \label{fig:exp_large_continuous_1Mb_histogram}
\end{figure}

Lastly, for the experiment in Section~\ref{sec:exp_tdt_loc}, we use the Selective SeqStep procedure \cite{barber2015} to combine the group p-values into a final set of discoveries. This procedure has one parameter, the threshold below which we reject p-values, which we set at $c = 0.5$.

\subsection{Robustness to computational phasing}
\label{sec:sim_phasing_robustness}

In practice, the Local Digital Twin Test relies on computationally
phased haplotypes, and we now verify that it is robust to imperfect
phasing. We carry out a simulation in the setting of
Section~\ref{sec:global_null_dtt_sim}, except that we split the chromosome into 30 groups of size approximately 2 Mb, and then apply the Local Digital Twin Test exactly as in Section~\ref{sec:sims_large_scale}. Note that in our synthetic population, we perfectly know the phase of both the parents and offspring.
To check our robustness to phasing, we consider either using (i) the true haplotypes or (ii) computationally phased haplotypes. For the latter, we input the trio {\em genotypes} to the \texttt{SHAPEIT} software \citep{Delaneau2012} (which accounts for the pedigree structure \citep{OConnell2014}) to obtain the computationally phased haplotypes. Lastly, we combine the group p-values using either the accumulation test \citep{li2017}, as in Section~\ref{sec:sims_large_scale}, or using the Selective SeqStep procedure \citep{barber2015} with parameter $c = 0.5$.

The results are reported in Figure~\ref{fig:sim_comp_phase}. We find that the power and FDR when using the computationally phased haplotypes are essentially identical to those when using the true haplotypes: the FDR is controlled at the nominal level while retaining the same power.

\begin{figure}[H]
  \centering
  \includegraphics[width = 6.5in]{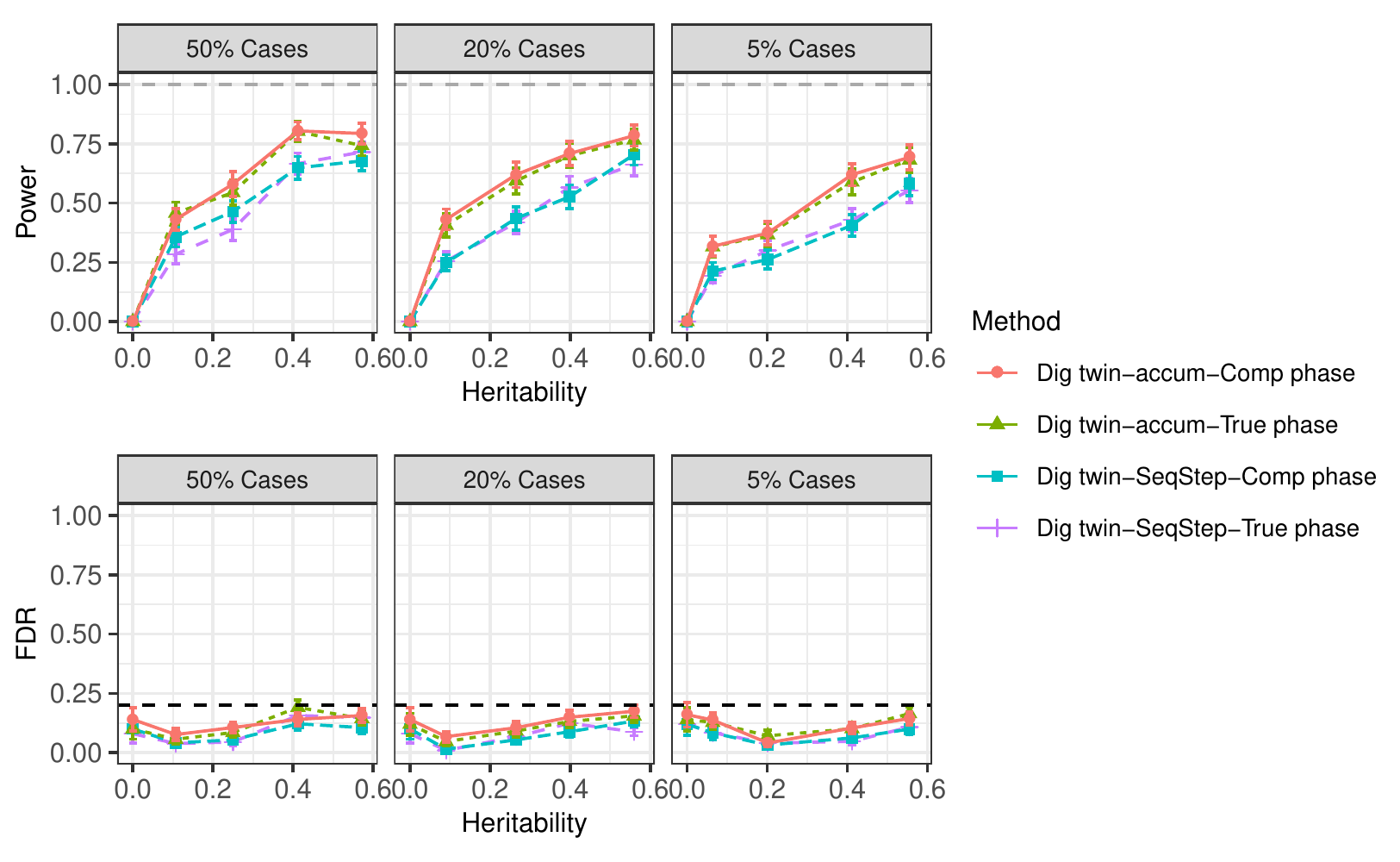}
  \caption{A comparison of the Local Digital Twin Test with perfect phasing to that with computational phasing. The solid grey line indicates the nominal FDR threshold.}
  \label{fig:sim_comp_phase}
\end{figure}

\subsection{Full-genome simulation with a continuous response}
\label{app:sim_full_gen_cts}

We carry out a simulation study analogous to that of Section~\ref{sec:sims_large_scale}, now with a continuous trait. As a benchmark, we substitute the QFAM family-based association test \cite{plink,fulker1999combined,abecasis2000general,purcell2007}, for the TDT, but otherwise proceed in an identical manner. We report the results in Figure~\ref{fig:exp_large_cts_causal}. In this case, the Digital Twin Test has much lower power than the QFAM-based benchmarks, except for high values of heritability, and all procedures empirically control the FDR.

\begin{figure}
  \centering
  \includegraphics[width=\linewidth]{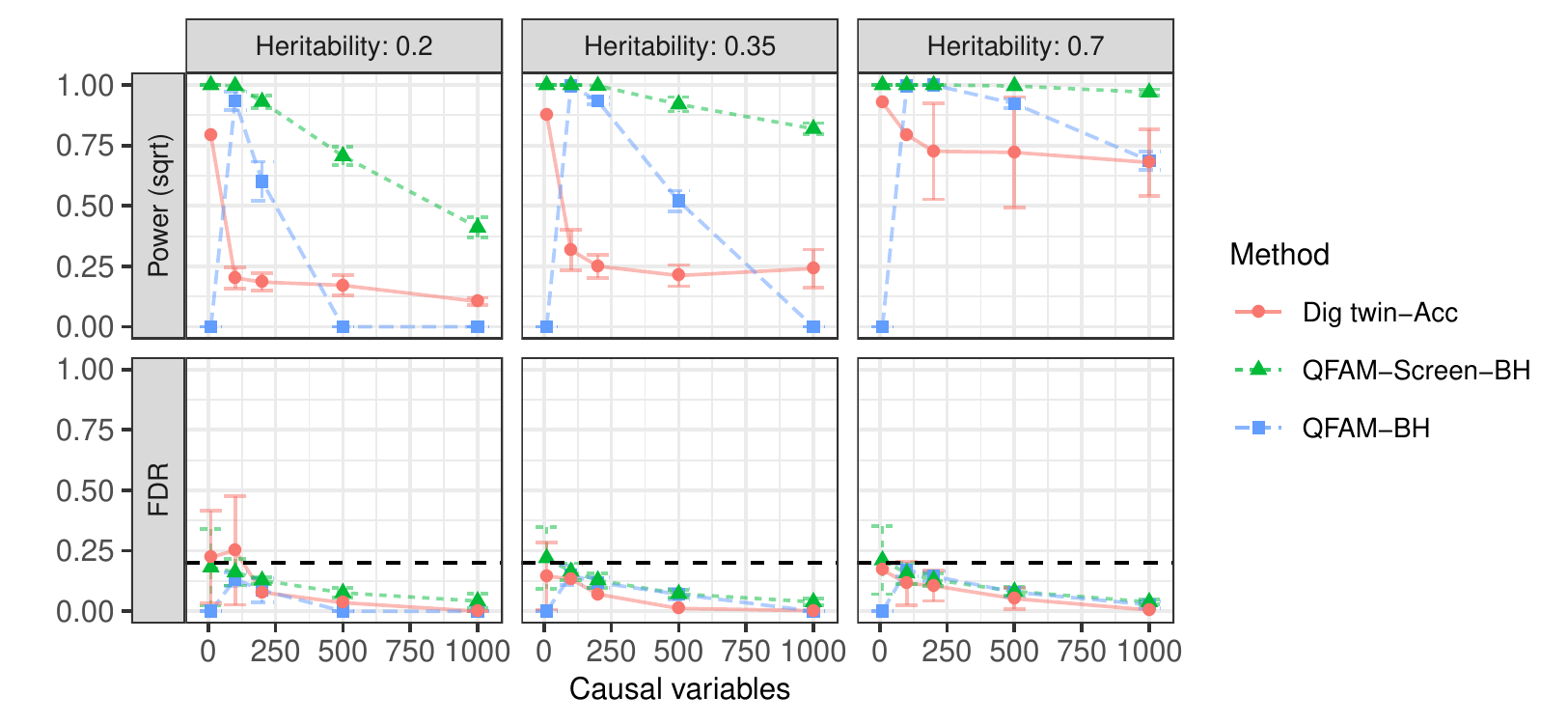}
  \caption{FDR and power of the Local Digital Twin Test in a large full-genome simulation with a continuous trait, compared to the performance of the QFAM procedure. Other details as in Figure~\ref{fig:exp_large_binary_causal}.}
  \label{fig:exp_large_cts_causal}
\end{figure}

\begin{figure}
  \centering
  \includegraphics[width=0.75\linewidth]{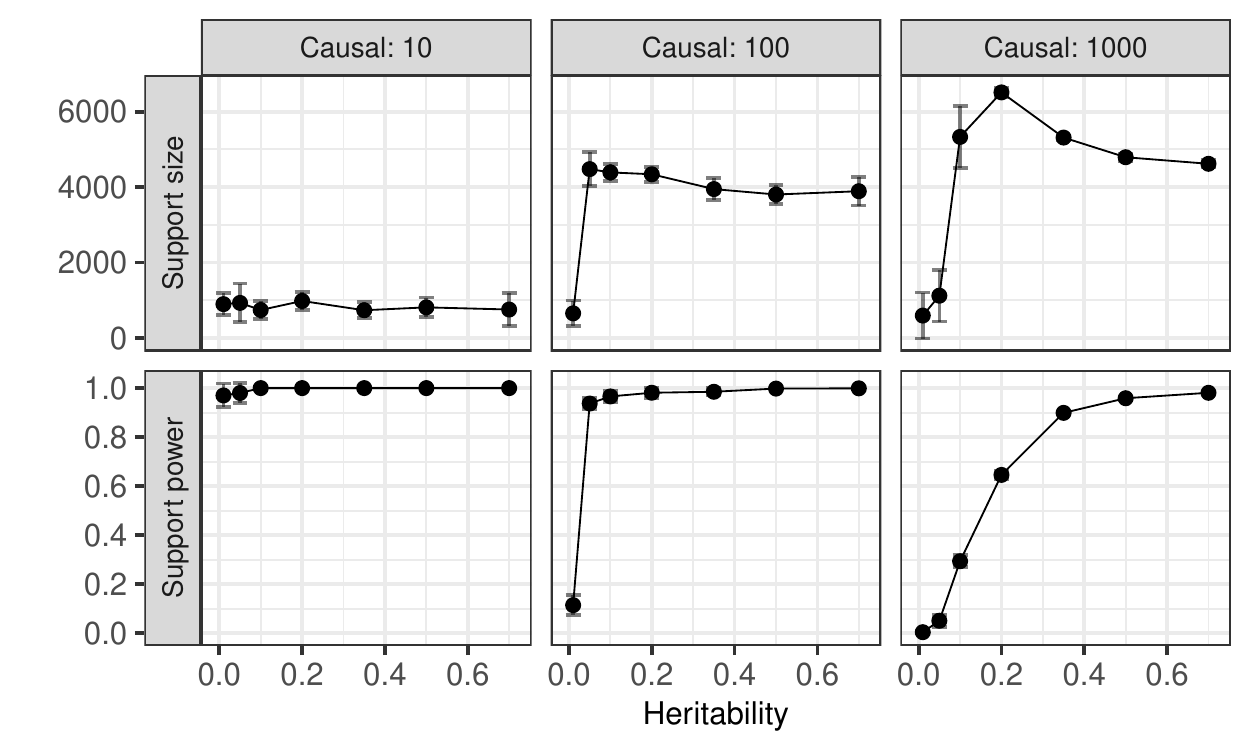}
  \caption{Properties of the regression coefficients $\hat{\beta}$ in the full-genome simulations with continuous trait. Other details as in Figure~\ref{fig:exp_large_binary_beta}.}
  \label{fig:exp_large_cts_beta}
\end{figure}

\end{document}